\documentclass[11pt]{article}
\usepackage{amsmath}
\usepackage{anysize}
\usepackage{latexsym}
\usepackage{amsthm, amssymb}
\usepackage{epsfig}
\usepackage{url}
\usepackage{pstricks}
\usepackage{tabls}
\usepackage{tabularx}
\usepackage{pst-node}
\usepackage{pst-plot}
\usepackage{appendix}
\usepackage{verbatim}
\usepackage{multirow,colortbl}
\usepackage{enumerate}
\usepackage{hyperref}
\hypersetup{colorlinks=false,linkbordercolor=green,citebordercolor=blue}

\marginsize{2.5cm}{2.5cm}{1.5cm}{2.5cm}

\pslongbox{MyFrame}{\psframebox}
{\MyFrame\begin{minipage}{#1}}%
{\end{minipage}\endMyFrame}

\theoremstyle{plain}
\newtheorem{theorem}{Theorem}
\newtheorem{lemma}{Lemma}

\newtheorem{claim}{Claim}
\newtheorem{corollary}{Corollary}

\theoremstyle{definition}
\newtheorem{definition}{Definition}

\theoremstyle{remark}
\newtheorem*{remark}{Remark}

\DeclareMathOperator{\suchthat}{\text{ }|\text{ }}

\begin{document}
\title{
On Randomized Memoryless Algorithms for the Weighted $k$-server Problem}
\author{Ashish Chiplunkar \qquad Sundar Vishwanathan\\
Department of Computer Science and Engineering\\
Indian Institute of Technology Bombay\\
Mumbai India\\
\texttt{\{ashishc, sundar\}@cse.iitb.ac.in}
}
\date{}
\maketitle

\begin{abstract}

The weighted $k$-server problem is a generalization of the $k$-server problem, wherein the cost of moving a server of weight $\beta_i$ through a distance $d$ is $\beta_i\cdot d$. On uniform metric spaces, this models caching with caches having different page replacement costs. We prove tight bounds on the performance of randomized memoryless algorithms for this problem on uniform metric spaces. We first prove that there is an $\alpha_k$-competitive memoryless algorithm for this problem, where $\alpha_k=\alpha_{k-1}^2+3\alpha_{k-1}+1$; $\alpha_1=1$. We complement this result by proving that no randomized memoryless algorithm can have a competitive ratio better than $\alpha_k$. 

To prove the upper bound of $\alpha_k$, we develop a framework to bound from above the competitive ratio of any randomized memoryless algorithm for this problem. The key technical contribution is a method for working with potential functions defined implicitly as the solution of a linear system. The result is robust in the sense that a small change in the probabilities used by the algorithm results in a small change in the upper bound on the competitive ratio. The above result has two important implications. Firstly, this yields an $\alpha_k$-competitive memoryless algorithm for the weighted $k$-server problem on uniform spaces. This is the first competitive algorithm for $k>2$, which is memoryless. For $k=2$, our algorithm agrees with the one given by Chrobak and Sgall \cite{ChrobakS04}. Secondly, this helps us prove that the Harmonic algorithm, which chooses probabilities in inverse proportion to weights, has a competitive ratio of $k\alpha_k$. 

The only known competitive algorithm for every $k$ before this work is a carefully crafted deterministic algorithm due to Fiat and Ricklin \cite{FiatR94}. This algorithm uses memory crucially, and their bound on its competitive ratio is $2^{4^k}$. Our algorithm is not only memoryless, but also has a considerably improved competitive ratio of $\alpha_k<1.6^{2^k}$. Further, the derandomization technique of Ben-David et al. \cite{Ben-DavidBKTW94} implies that there exists a deterministic algorithm for this problem with competitive ratio $\alpha_k^2<2.56^{2^k}$.
\end{abstract}

\section{Introduction}

The $k$-server problem of Manasse et al. \cite{ManasseMS88} is, arguably, the most extensively studied problem in the online setting. The large body of research around this problem is summarized in a beautiful survey by Koutsoupias \cite{Koutsoupias09}. In this problem, $k$ servers occupy points in a metric space. An adversary presents a sequence of requests, each of which is a point in the metric space. To serve the current request, the algorithm moves one of the servers to the requested point, incurring a cost equal to the distance traveled by the server. In the online model, an algorithm is required to serve the current request before the next request is revealed. A (randomized) online algorithm is said to be $c$-\textit{competitive} against an adversary, if it produces a solution, whose (expected) cost is at most $c$ times the cost of the solution produced by the adversary. 

A generalization of the $k$-server problem, proposed by Fiat and Ricklin \cite{FiatR94}, and called the weighted $k$-server problem, associates a weight with each server. The cost incurred in moving a server is equal to the product of its weight and the distance traveled. Introducing weights adds a new dimension to the $k$-server problem and presents new challenges. While a $(2k-1)$-competitive algorithm is known for the $k$-server problem \cite{KoutsoupiasP95}, the only competitive algorithms known for the weighted $k$-server problem are for uniform spaces \cite{FiatR94}, and for $k=2$ \cite{Sitters14}. On uniform spaces, this problem models caching with different types of caches, each having a different page replacement cost. Fiat and Ricklin \cite{FiatR94} point out the practical significance of such caches in optimizing both the overall write time, as well as the chip area occupied.

A randomized algorithm for the weighted $k$-server problem is said to be \textit{memoryless} if its behavior on a request is completely determined by the pairwise distances between the $k$ points occupied by its servers and the requested point. In other words, a memoryless algorithm for the weighted $k$-server problem with a given set of weights is specified by a function, which maps the ${k+1}\choose2$ distances to a probability distribution on the servers. In particular, on uniform metric spaces, a memoryless algorithm is completely specified by a probability distribution $p$ on the servers, where $p_i$ is the probability by which the $i^{\text{\tiny{th}}}$ server is shifted to the requested point, if that point is not already occupied by some server. The Harmonic algorithm is a memoryless algorithm, which moves the servers with probabilities inversely proportional to their weights.

For online problems modeling certain practical problems like caching, it is imperative that decisions are taken instantaneously. Ideally, we would like the algorithm to be memoryless. For the $k$-server problem, the Harmonic algorithm is known to be $O(k2^k)$-competitive on any metric space \cite{Grove91,BartalG00}. Additionally, Coppersmith et al. \cite{CoppersmithDRS93} proved that on \textit{resistive} metric spaces, there exists a $k$-competitive memoryless algorithm, in which the probabilities of moving the servers are determined by the \textit{resistive inverse} of the metric space. It hence came as a surprise when Chrobak and Sgall \cite{ChrobakS04} proved that no memoryless algorithm with a finite competitive ratio exists, even for the weighted $2$-server problem on the line metric (which is, in fact, resistive). Among other nice results in the same paper, Chrobak and Sgall \cite{ChrobakS04} give the only known competitive memoryless algorithm for uniform spaces: a $5$-competitive algorithm for $2$ servers, which they prove is optimal. We generalize their bounds and prove the following theorems. 

\begin{theorem}\label{thm_main_ub}
For every $k$, there exists an $\alpha_k$-competitive memoryless algorithm for the weighted $k$-server problem on uniform metric spaces against an online adaptive adversary, where $\alpha_k$ satisfies the recurrence: $\alpha_k=\alpha_{k-1}^2+3\alpha_{k-1}+1$ for $k>1$, and $\alpha_1=1$.
\end{theorem}

\begin{theorem}\label{thm_main_lb}
There does not exist a memoryless algorithm for the weighted $k$-server problem on uniform metric spaces with competitive ratio less than $\alpha_k$, for any $k$, against an online adaptive adversary.
\end{theorem}

In order to establish Theorem \ref{thm_main_ub}, we prove a more general result. Given server weights $\beta=(\beta_1,\ldots,\beta_k)$, and a probability distribution $p$ on the servers used by an algorithm, we derive an upper bound $\tilde{\alpha}(\beta,p)$ on the competitive ratio, as a function of $\beta$ and $p$. Given $\beta$, we use this result to identify a probability distribution $p$, such that the competitive ratio is at most $\alpha_k$. As a by-product of this more general result, we also derive that the Harmonic algorithm is $(k\alpha_k)$-competitive, for any $\beta$, against an online adaptive adversary. For $k=2$, we get $\alpha_2=5$, and our result matches that of Chrobak and Sgall \cite{ChrobakS04}.

Towards proving Theorem \ref{thm_main_lb}, we first prove that the upper bound of $\tilde{\alpha}(\beta,p)$ is tight. Specifically, we prove that if the \textit{separation} $\min_i\beta_{i+1}/\beta_i$ between the weights is sufficiently large, then there exists an online adaptive adversary, which forces the algorithm using the probability distribution $p$ to perform almost $\tilde{\alpha}(\beta,p)$ times worse. It is interesting to note that we leverage the machinery developed to prove the upper bound, to prove this lower bound too; we use the same potentials in a different avatar. We then prove that with weights $1,r,r^2,\ldots,r^{k-1}$, for a sufficiently large $r$, $\inf_p\tilde{\alpha}(\beta,p)$ can be forced to be arbitrarily close to $\alpha_k$.

The main difficulty in analyzing algorithms for this problem stems from the inability to describe suitable potential functions explicitly. We formulate a set of linear inequalities that the potentials must satisfy, where the co-efficients involved in the inequalities depend on the probabilities and the weights. This by itself has been done before; see for example \cite{BartalCL98}. However, the rest of the work is very different. We then show that the point, at which a certain carefully chosen subset of the linear inequalities is tight, is feasible. Our work indicates that the potentials given by this point are complicated rational functions of the probabilities and weights, and describing them seems hopeless, even for $k=4$. Our key technical contribution is a framework to work with potential functions defined implicitly, as the solution of a linear system.

Theorem \ref{thm_main_ub} also has the following consequence. Together with the derandomization result by Ben-David et al. \cite{Ben-DavidBKTW94}, it implies the existence of a deterministic algorithm, for the weighted $k$-server problem on uniform spaces, with competitive ratio $\alpha_k^2$. It can be easily proved that $\alpha_k<1.6^{2^k}$ and thus, we have an upper bound of $2.56^{2^k}$. This is significantly better than the earlier bound on the deterministic competitive ratio by Fiat and Ricklin \cite{FiatR94}, which was more than $2^{4^k}$.

\section{Preliminaries and Techniques}\label{sec_preliminaries}

Let $\beta=(\beta_1,\ldots,\beta_k)$ be the weights of the servers in an instance of the weighted $k$-server problem. Consider a memoryless algorithm that, in response to a request on a point not already occupied by a server, moves the $i^{\text{\tiny{th}}}$ server with probability $p_i$. We derive an upper bound on its competitive ratio, as a function of $\beta$ and $p=(p_1,\ldots,p_k)$. Note that whenever a point not occupied by the algorithm's servers is requested, the expected cost incurred by the algorithm is $\sum_{j=1}^kp_j\beta_j$.

\subsection{Potential functions}

In this paper, we design algorithms against an online adaptive adversary \cite{Ben-DavidBKTW94}. An online adaptive adversary observes the behavior of the algorithm on the previous requests, generates the next request, and immediately serves it. The traditional method for analyzing an online algorithm is to associate a \textit{potential} with each \textit{state}, determined by the positions of the adversary's and algorithm's servers, such that
\begin{enumerate}
\item When the adversary moves, the increase in the potential is at most $\alpha$ times the cost incurred by it.
\item When the algorithm moves, the decrease in the potential is at least as much as the cost incurred by the algorithm.
\end{enumerate}
We think of each request being first served by the adversary, and then by the algorithm. A standard telescoping argument implies that the competitive ratio is then bounded from above by $\alpha$.

In our case, we define the states as follows. At any point of time, let $a_i$ (resp. $s_i$) denote the position of the adversary's (resp. algorithm's) $i^{\text{\tiny{th}}}$ server. We identify our state with the set $S=\{i\suchthat a_i=s_i\}\subseteq[k]$. We denote by $\phi_S$ the potential we associate with state $S$. We assume, without loss of generality, that the adversary never requests a point occupied by one of algorithm's servers, and that the adversary moves its servers only to serve requests. Suppose that at some point of time the state is $S$, and the adversary moves its $i^{\text{\tiny{th}}}$ server, incurring a cost $\beta_i$. If $i\notin S$, then the state does not change, while if $i\in S$ the state changes to $S\setminus\{i\}$. In order to prove $\alpha$-competitiveness it is sufficient to have potentials satisfying
\begin{equation}\label{eqn_phi_adv}
\phi_{S\setminus\{i\}}-\phi_S\leq\beta_i\cdot\alpha\text{ for every }S\text{ and }i\in S
\end{equation}

Suppose that the current state is $S$, and it is the algorithm's turn to serve the request. The request must be $a_i$ for some $i\notin S$. If the algorithm moves its $i^{\text{\tiny{th}}}$ server, the new state is $S\cup\{i\}$. This happens with probability $p_i$, and the decrease in potential is $\phi_S-\phi_{S\cup\{i\}}$. Else if the algorithm moves its $j^{\text{\tiny{th}}}$ server for some $j\in S$, the new state is $S\setminus\{j\}$. This happens with probability $p_j$, and the decrease in potential is $\phi_S-\phi_{S\setminus\{j\}}$. Finally, if the algorithm moves its $j^{\text{\tiny{th}}}$ server for some $j\notin S$ and $j\neq i$, there is no change in the state, and hence the potential. We want the expected decrease in potential to be at least the expected cost incurred by the algorithm. Thus, we need
\begin{equation}\label{eqn_phi_alg}
p_i(\phi_S-\phi_{S\cup\{i\}})-\sum_{j\in S}p_j(\phi_{S\setminus\{j\}}-\phi_S)\geq\sum_{j=1}^kp_j\beta_j\text{ for every }S\text{ and }i\notin S
\end{equation}

\subsection{A Linear Program and a choice of an Extreme Point}\label{subsec_lp_extremept}

Among the set of potentials $\phi_S$, for each $S\subseteq[k]$, satisfying (\ref{eqn_phi_alg}), we wish to pick one to minimize $\alpha$, which is bounded from below due to (\ref{eqn_phi_adv}). The conditions (\ref{eqn_phi_alg}) define a polyhedron in $\mathbb{R}^{2^k}$. Note that the right hand side of each constraint in (\ref{eqn_phi_alg}) is constant. We assume that $\phi_{\emptyset}$, the potential of the empty set, is $0$.

To simplify calculations and to facilitate an inductive approach, we introduce, with foresight, a change of variables. Let us replace the variables $(\phi_S)_{S\subseteq[k]}$ by the variables $(\varphi_S)_{S\subseteq[k]}$ such that $\phi_S=-(\sum_{j=1}^kp_j\beta_j)\varphi_S$. With this substitution, we have the following optimization problem.
\begin{center}
Minimize $\alpha$ subject to
\end{center}
For every $S$ and $i\in S$,
\begin{equation}\label{eqn_varphi_adv}
\left(\sum_{j=1}^kp_j\beta_j\right)\frac{\varphi_S-\varphi_{S\setminus\{i\}}}{\beta_i}\leq\alpha
\end{equation}
For every $S$ and $i\notin S$,
\begin{equation}\label{eqn_varphi_alg}
p_i\left(\varphi_{S\cup\{i\}}-\varphi_S\right)-\sum_{j\in S}p_j\left(\varphi_S-\varphi_{S\setminus\{j\}}\right)\geq1
\end{equation}
\begin{equation}\label{eqn_varphi_0}
\varphi_{\emptyset}=0
\end{equation}
Note that the objective value of this problem does not change when all the $p_j$'s are scaled by a positive constant, and the feasible space merely gets scaled by the inverse of that constant. Henceforth, for convenience, we will ignore the fact that $p_1,\ldots,p_k$ sum to $1$. We may think of $p_1,\ldots,p_k$ as the relative frequencies of moving the respective servers, 
and therefore, the algorithm moves the $i^{\text{\tiny{th}}}$ server with probability $p_i/(\sum_{j=1}^kp_j)$. 

Our task is to establish the existence of one feasible point of the linear program given by (\ref{eqn_varphi_adv}), (\ref{eqn_varphi_alg}), (\ref{eqn_varphi_0}) with the required bound on the competitive ratio. Recall that linear programming theory says the optimum must be attained at an extreme point. We guess a subset of $2^k$ linearly independent constraints among (\ref{eqn_varphi_alg}) which will be satisfied with equality. This forms the implicit description of the potentials. Assume without loss of generality, that $p_1\geq\cdots\geq p_k$. Let $\varphi_S=\varphi_S(p)$ for all $S$ be the solution of the following linear system of equations:
\begin{equation}\label{eqn_varphi_tight}
p_i\left(\varphi_{S\cup\{i\}}-\varphi_S\right)-\sum_{j\in S}p_j\left(\varphi_S-\varphi_{S\setminus\{j\}}\right)=1
\end{equation}
for every $S\neq[k]$ and $i$: the smallest integer not in $S$, assuming $\varphi_{\emptyset}=\varphi_{\emptyset}(p)=0$.
Let $\alpha(p)=\max_{S,i\in S}\left(\sum_{j=1}^kp_j\beta_j\right)\frac{\varphi_S(p)-\varphi_{S\setminus\{i\}}(p)}{\beta_i}$. We need to prove that the quantities $\varphi_S(p)$, as defined above, constitute a feasible point, that is, satisfy all the remaining constraints in (\ref{eqn_varphi_alg}). We will then prove an upper bound on the value of the objective function $\alpha(p)$.

\subsection{Checking Feasibility: The Gauss-Seidel Trick}\label{sec_gaussseidel}

The na\"\i ve way to check feasibility is to determine each $\varphi_S(\cdot)$ explicitly, substitute in (\ref{eqn_varphi_alg}), and verify that the constraints are satisfied for every $p$. However, these functions tend to be more and more complicated as $k$ grows, and it is hopeless to find the closed form expressions, even when $k=4$. We therefore resort to the following indirect way.

Suppose we want to prove that the solution $x^*\in\mathbb{R}^n$ of the system $Ax=b$ satisfies $c^{\top}x^*\leq d$. The Gauss-Seidel iterative procedure in numerical computation to compute $x^*$ is as follows. Write $A$ as $L_*+U$, where $L_*$ consists of the diagonal and the lower triangular part of $A$, and $U$ consists of the upper triangular part. Choose an initial point $x^0$, and for $i$ going from $1$ to $\infty$, calculate $x^i=L_*^{-1}(b-Ux^{i-1})$. In other words, in every iteration, the $j^{\text{\tiny{th}}}$ coordinate is computed using the $j^{\text{\tiny{th}}}$ equality in the system $Ax=b$, and the latest values of other coordinates. The coordinates are computed in a fixed order in all iterations. Under certain sufficiency conditions on $A$, which imply that $L^*$ is invertible, the sequence $(x^i)$ converges to $x^*$.

Our technique for proving $c^{\top}x^*\leq d$ is as follows. With a suitable choice of the initial point $x^0$, we prove that $c^{\top}x^0\leq d$, and that for all $i$, $c^{\top}x^{i-1}\leq d$ implies $c^{\top}x^i\leq d$. This proves that the entire sequence $(x^i)$ satisfies the constraint $c^{\top}x\leq d$. Further, since the constraint defines a closed subset of $\mathbb{R}^n$, the limit point $x^*$ also satisfies the constraint. We will call this trick the Gauss-Seidel trick for feasibility checking.

Two sufficient conditions for the Gauss-Seidel iterations to converge are that the system be strictly diagonally dominated, or irreducibly diagonally dominated.\footnote{\url{http://en.wikipedia.org/wiki/Gauss-Seidel_method}} In our case, the system given by (\ref{eqn_varphi_tight}) is diagonally dominated, but it is neither strictly diagonally dominated, nor irreducibly diagonally dominated. Hence, the Gauss-Seidel trick does not apply directly. To deal with this and other minor technical issues, we will define $\varphi_S$ inductively, using certain other functions $f_S$ of the probabilities. For a fixed probability distribution, these functions $f_S$ will be the solutions of a strictly diagonally dominated linear system. We will use the Gauss-Seidel trick to prove certain linear inequalities involving these functions, which will imply the inequalities in (\ref{eqn_varphi_alg}).

\section{Proof of the Upper Bound}\label{sec_upperbound}

\subsection{Defining the Potentials}\label{subsec_def_potentials}

Let $\mathbb{R}(X)=\mathbb{R}(X_1,X_2,\ldots)$ denote the field of rational expressions over the countably infinite set of indeterminates $\{X_1,X_2,\ldots\}$. Recall that any element of this field is a ratio of a multivariate polynomial to another non-zero multivariate polynomial. By definition, a polynomial is a real combination of finitely many monomials. Hence, any rational expression in $\mathbb{R}(X)$ involves only finitely many indeterminates. Given any $\varphi\in\mathbb{R}(X)$, let $n$ be the largest integer such that $X_n$ appears in $\varphi$. Then for $q=(q_1,q_2,\ldots,q_m)$, the evaluation $\varphi(q)\in\mathbb{R}$ for the substitution $X_i=q_i$ is well defined, as long as $m\geq n$ and the denominator of $\varphi$ does not vanish at $q$.

We will now formally define the functions $\varphi_S\in\mathbb{R}(X)$,
one for each finite subset $S\subseteq\mathbb{N}$. 
If $n$ is the largest integer in $S$, then the rational expression $\varphi_S$ will involve the indeterminates $X_1,\ldots,X_n$ only, and its evaluation will be well defined in the open positive orthant of $\mathbb{R}^n$.

For a finite $S\subseteq\mathbb{N}$ and $i\in S$, define the rational function $\mathbf{I}^{S\setminus\{i\}}_S\in\mathbb{R}(X)$ to be $X_i(\varphi_S-\varphi_{S\setminus\{i\}})$.
The function $\varphi_S$ is to be thought of as an electric potential applied at $S$, with the sets $S$ and $S\setminus\{i\}$ connected by a conductance $X_i$. Therefore $\mathbf{I}^{S\setminus\{i\}}_S$ can be viewed as the current flowing from $S$ to $S\setminus\{i\}$. Note that the potentials and currents will satisfy the Kirchhoff's voltage law but not the current law.

We define $\varphi_S$ by induction on the largest integer $n$ in $S$. $\varphi_{\emptyset}$ is defined to be identically zero. Having defined $\varphi_T$ for each $T\subseteq[n-1]$ (and hence, $\mathbf{I}^{[n-1]\setminus\{j\}}_{[n-1]}$ for each $j\in[n-1]$), for $S\subseteq[n]$ such that $n\in S$, we define
\begin{equation}\label{eqn_def_varphi}
\varphi_S=\varphi_{S\setminus\{n\}}+\frac{f_S}{X_n}
\end{equation}
where the rational functions $f_S\in\mathbb{R}(X)$ satisfy the following equations.
\begin{equation}\label{eqn_f2}
f_{[n]}=1+\sum_{j=1}^{n-1}\mathbf{I}^{[n-1]\setminus\{j\}}_{[n-1]}
\end{equation}
and for $S\neq[n]$, with $i<n$, the smallest integer not in $S$,
\begin{equation}\label{eqn_f1}
\left(X_i+\sum_{j\in S}X_j\right)f_S=X_if_{S\cup\{i\}}+\sum_{j\in S\setminus\{n\}}X_jf_{S\setminus\{j\}}
\end{equation}
We claim that the linear system given by the equations (\ref{eqn_f1}) and (\ref{eqn_f2}) in the variables $\{f_S\suchthat S\subseteq[n]\text{, }n\in S\}$ has a unique solution in $\mathbb{R}(X)$. To see this, substitute arbitrary positive values for $X_1,\ldots,X_n$, and observe that the coefficient matrix is strictly diagonally dominant. This ensures that the determinant of the coefficient matrix is not identically zero. Furthermore, this also guarantees that the evaluation $f_S(p)$ is well defined, for $p=(p_1\ldots,p_n)\in(\mathbb{R}_{>0})^n$. Moreover, the Gauss-Seidel procedure converges to the solution when started from any point.

It will be useful to have an expression for the currents in terms of the functions $f_S$, which we derive next. First, note that $\mathbf{I}^{S\setminus\{n\}}_S=f_S$. Further, for $i\in S$, $i<n$ we have
\begin{eqnarray}\label{eqn_current_f}
\mathbf{I}^{S\setminus\{i\}}_S & = & X_i\left(\varphi_S-\varphi_{S\setminus\{i\}}\right)\nonumber\\
 & = & X_i\left[\left(\varphi_{S\setminus\{n\}}+\frac{f_S}{X_n}\right)-\left(\varphi_{S\setminus\{i,n\}}+\frac{f_{S\setminus\{i\}}}{X_n}\right)\right]\nonumber\\
 & = & X_i\left(\varphi_{S\setminus\{n\}}-\varphi_{S\setminus\{i,n\}}\right)+\frac{X_i}{X_n}\left(f_S-f_{S\setminus\{i\}}\right)\nonumber\\
 & = & \mathbf{I}^{S\setminus\{i,n\}}_{S\setminus\{n\}}+\frac{X_i}{X_n}\left(f_S-f_{S\setminus\{i\}}\right)
\end{eqnarray}
We note the following facts, which can be easily proved by induction. For any constant $c\in\mathbb{R}$,
\begin{enumerate}
\item $\varphi_S(cX)=\varphi(cX_1,cX_2,\ldots)=\varphi_S(X)/c$.
\item $\mathbf{I}^{S\setminus\{i\}}_S(cX)=\mathbf{I}^{S\setminus\{i\}}_S(X)$. In particular, $f_S(cX)=f_S(X)$.
\end{enumerate}

While (\ref{eqn_def_varphi}), (\ref{eqn_f2}), (\ref{eqn_f1}) may be taken as an independent definition, we will now show that this definition of $\varphi_S$ coincides with the equation (\ref{eqn_varphi_tight}) in Section \ref{subsec_lp_extremept}.

\begin{lemma}\label{lem_eqn_varphi_tight}
Let $S$ be a finite subset of $\mathbb{N}$ and let $i$ be the smallest integer not in $S$. Then
\[X_i\left(\varphi_{S\cup\{i\}}-\varphi_S\right)=1+\sum_{j\in S}X_j\left(\varphi_S-\varphi_{S\setminus\{j\}}\right)\]
or equivalently
\[\mathbf{I}^S_{S\cup\{i\}}=1+\sum_{j\in S}\mathbf{I}^{S\setminus\{j\}}_S\]
\end{lemma}

\begin{proof}
The claim is true for $S=\emptyset$, since $\mathbf{I}^{\emptyset}_{\{1\}}=f_{[1]}=1$, by equation (\ref{eqn_f2}). We prove by induction on the largest integer $n$ in $S$. First, consider the case when $S=[n]$. Then $i=n+1$, and we have $\mathbf{I}^S_{S\cup\{i\}}=f_{[n+1]}=1+\sum_{j=1}^{n}\mathbf{I}^{[n]\setminus\{j\}}_{[n]}$, where the second equality is given by equation (\ref{eqn_f2}). In particular, the claim holds for $n=1$.

Now for the inductive step, assume $n>1$, and $S\subsetneq[n]$. Hence, $i<n$. We have from equation (\ref{eqn_current_f})
\begin{equation}\label{eqn_subs}
\mathbf{I}^S_{S\cup\{i\}}=\mathbf{I}^{S\setminus\{n\}}_{S\cup\{i\}\setminus\{n\}}+\frac{X_i}{X_n}\left(f_{S\cup\{i\}}-f_S\right)
\end{equation}
But $i$ is also the smallest integer not in $S\setminus\{n\}$, and by the induction hypothesis, we have
\[\mathbf{I}^{S\setminus\{n\}}_{S\cup\{i\}\setminus\{n\}}=1+\sum_{j\in S\setminus\{n\}}\mathbf{I}^{S\setminus\{j,n\}}_{S\setminus\{n\}}\]
Further, rearranging equation (\ref{eqn_f1}), we get
\[X_i\left(f_{S\cup\{i\}}-f_S\right)=X_nf_S+\sum_{j\in S\setminus\{n\}}X_j\left(f_S-f_{S\setminus\{j\}}\right)\]
Substituting in equation (\ref{eqn_subs}), and again using equation (\ref{eqn_current_f}), we get
\[\mathbf{I}^S_{S\cup\{i\}}=1+f_S+\sum_{j\in S\setminus\{n\}}\left[\mathbf{I}^{S\setminus\{j,n\}}_{S\setminus\{n\}}+\frac{X_j}{X_n}(f_S-f_{S\setminus\{j\}})\right]=1+\mathbf{I}^{S\setminus\{n\}}_S+\sum_{j\in S\setminus\{n\}}\mathbf{I}^{S\setminus\{j\}}_S=1+\sum_{j\in S}\mathbf{I}^{S\setminus\{j\}}_S\]
\end{proof}

\begin{remark}
The definitions of the rational functions $\varphi_S$, $\mathbf{I}^S_{S\cup\{i\}}$, $f_S$ are all independent of $k$, the number of servers under consideration. However, in the analysis of the weighted $k$-server problem, the potential function is to be defined for sets $S\subseteq[k]$ only. Hence, we use the functions $\varphi_S$, $\mathbf{I}^S_{S\cup\{i\}}$, $f_S$ for $S\subseteq[k]$ and $i\leq k$ only. These involve the indeterminates $X_1,\ldots,X_k$ only. For a randomized memoryless algorithm, given by a probability distribution $p=(p_1,\ldots,p_k)$ on the servers, we evaluate these functions by the substitution $X_i=p_i$, and use the evaluations in our analysis.
\end{remark}

\subsection{Proving Feasibility}

Towards proving Theorem \ref{thm_main_ub}, our first goal is to prove that the potentials satisfy the constraints given by (\ref{eqn_varphi_alg}). Towards this, we first prove suitable inequalities involving the evaluations $f_S(p)$, and then use induction and (\ref{eqn_current_f}). The inequalities that we need are given by the following lemma, which essentially says that the quantity $p_i\left(f_{S\cup\{i\}}(p)-f_S(p)\right)$ is monotone with respect to $i$, for a fixed $S$ and $p$.

\begin{lemma}\label{lemma_f}
Let $S\subseteq[n]$ be a set containing $n$. Suppose $p=(p_1,\ldots,p_n)$ with $p_1\geq\cdots\geq p_n>0$. Then for any $i,i'\notin S$, $i<i'<n$, we have
\[p_i\left(f_{S\cup\{i\}}(p)-f_S(p)\right)\leq p_{i'}\left(f_{S\cup\{i'\}}(p)-f_S(p)\right)\]
\end{lemma}

The proof of this lemma is the technical heart of the upper bound. It is in this proof that we use the Gauss-Seidel trick. We prove that the claim is true after every iteration of the Gauss-Seidel procedure, when started from an appropriately chosen point. This lemma enables us to prove the following current monotonicity property.

\begin{lemma}[Monotonicity of currents]\label{lem_currents}
Let $S$ be a finite subset of $\mathbb{N}$, $i,j\notin S$, $i<j$ and $S\cup\{i,j\}\subseteq[n]$. Suppose $p=(p_1,\ldots,p_n)$ with $p_1\geq\cdots\geq p_n>0$. Then we have
\[\mathbf{I}^S_{S\cup\{i\}}(p)\leq\mathbf{I}^S_{S\cup\{j\}}(p)\]
\end{lemma}

We defer the proofs of the above two claims to the Appendix. The following feasibility lemma, which states that the constraints (\ref{eqn_varphi_alg}) are satisfied, is immediate from Lemmas \ref{lem_eqn_varphi_tight} and \ref{lem_currents}.

\begin{lemma}[Feasibility]\label{lem_eqn_varphi}
Let $S$ be a finite subset of $\mathbb{N}$, $i\notin S$, and $S\cup\{i\}\subseteq[n]$. Suppose $p=(p_1,\ldots,p_n)$ with $p_1\geq\cdots\geq p_n>0$. Then we have
\[p_i\left(\varphi_{S\cup\{i\}}(p)-\varphi_S(p)\right)\geq1+\sum_{j\in S}p_j\left(\varphi_S(p)-\varphi_{S\setminus\{j\}}(p)\right)\]
or equivalently
\[\mathbf{I}^S_{S\cup\{i\}}(p)\geq1+\sum_{j\in S}\mathbf{I}^{S\setminus\{j\}}_S(p)\]
with equality if $i$ is the smallest integer not in $S$.
\end{lemma}

\subsection{Bounding the Objective Function}

Recall that in Section \ref{subsec_lp_extremept}, we mentioned that the randomized memoryless algorithm, which uses the probability distribution $p=(p_1,\ldots,p_k)$ with $p_1\geq\cdots\geq p_k$, for server weights $\beta_1,\ldots,\beta_k$, has a competitive ratio bounded by 
\[\left(\sum_{j=1}^kp_j\beta_j\right)\max_{S,i\in S}\frac{\varphi_S(p)-\varphi_{S\setminus\{i\}}(p)}{\beta_i}=\left(\sum_{j=1}^kp_j\beta_j\right)\max_{S,i\in S}\frac{\mathbf{I}^{S\setminus\{i\}}_S(p)}{p_i\beta_i}\]
Given $\beta_1,\ldots,\beta_k$, we would like to choose a probability distribution $p$ that minimizes this. How this can be done is unclear due to the presence of the rational function $\mathbf{I}^{S\setminus\{i\}}_S$. However, we will show that each current $\mathbf{I}^{S\setminus\{i\}}_S(p)$ is bounded from above by constants (which depend on $S$ and $i$ but not on $p$). Towards proving this, the key property we need is that for any $p$, $S\mapsto\varphi_S(p)$ is a supermodular set function.

\begin{lemma}[Supermodularity]\label{lem_supermodularity}
Let $S$ be a finite subset of $\mathbb{N}$, $i,j\notin S$, and $S\cup\{i,j\}\subseteq[n]$. Suppose $p=(p_1,\ldots,p_n)$ with $p_1\geq\cdots\geq p_n>0$. Then we have
\[\varphi_{S\cup\{i\}}(p)+\varphi_{S\cup\{j\}}(p)\leq\varphi_{S\cup\{i,j\}}(p)+\varphi_S(p)\]
Thus for any fixed $p$, the function mapping a set $S$ to $\varphi_S(p)$ is supermodular, and we have
\[\mathbf{I}^{S'}_{S'\cup\{i\}}(p)=p_i\left(\varphi_{S'\cup\{i\}}(p)-\varphi_{S'}(p)\right)\leq p_i\left(\varphi_{S\cup\{i\}}(p)-\varphi_S(p)\right)=\mathbf{I}^S_{S\cup\{i\}}(p)\]
whenever $S'\subseteq S\subseteq[n]$ and $i\in[n]\setminus S$.
\end{lemma}

We defer the proof to the Appendix. As mentioned earlier, this lemma enables us to prove that each current is bounded from above by a constant independent of $p$. We will define one such constant for each finite subset of $\mathbb{N}$, by induction on the following enumeration of the finite subsets of $\mathbb{N}$, called the \textit{colex order}.

\begin{definition}\label{def_colex}
Let $S$ and $T$ be finite subsets of $\mathbb{N}$. We say that $S$ precedes $T$ in the \textit{colex order} if there exists $i\in T\setminus S$ such that $S$ and $T$ agree on membership of integers greater than $i$.
\end{definition}

For example, the first few sets in the colex order are $\emptyset$, $\{1\}$, $\{2\}$, $\{1,2\}$, $\{3\}$, $\{1,3\}$, $\{2,3\}$, $\{1,2,3\}$, etc. Given a set $S$, the next set is obtained by the including $i$, the smallest number not in $S$, and removing all the numbers $1,\ldots,i-1$ from $S$. Similarly, the set just before $S$ is obtained by removing from $S$ the smallest number $j$ in it, and putting in all the smaller numbers $1,\ldots,j-1$. We will refer to the colex order in the forthcoming inductive definitions and proofs.

\begin{definition}\label{def_C}
For each finite $S\subseteq\mathbb{N}$ define $C_S$ using the following recurrence. $C_{\emptyset}=1$ and if $S\neq\emptyset$, then $C_S=1+\sum_{j\in S}C_{S\setminus\{j\}\cup[j-1]}$. For each $n\in\mathbb{N}$ define $\alpha_n$ as $\alpha_n=\alpha_{n-1}^2+3\alpha_{n-1}+1$, where $\alpha_0=0$.
\end{definition}

Note that the above definition is valid because $S\setminus\{j\}\cup[j-1]$ precedes $S$ in the colex order, for any $j\in S$. The bounds on the currents are given by the following lemma.

\begin{lemma}[Boundedness of currents]\label{lem_currents_abs_bound}
For every finite set $S\subseteq[n]$ and for all $p=(p_1,\ldots,p_n)$ with $p_1\geq\cdots\geq p_n>0$, $\mathbf{I}^S_{S\cup\{i\}}(p)\leq C_S$, where $i$ is the smallest integer not in $S$.
\end{lemma}

\begin{proof}
We induct on the position of $S$ in the colex order. For the base case, when $S=\emptyset$ and $i=1$, we have $\mathbf{I}^{\emptyset}_{\{1\}}(p)=f_{\{1\}}(p)=1=C_{\emptyset}$ for all $p$, by equation (\ref{eqn_f2}). For the inductive case, assume that the claim holds for all finite subsets of $\mathbb{N}$ which precede a set $S$ in the colex order. Let $i$ be the smallest integer not in $S$. Then by Lemma \ref{lem_eqn_varphi_tight}, we have
\[\mathbf{I}^S_{S\cup\{i\}}(p)=1+\sum_{j\in S}\mathbf{I}^{S\setminus\{j\}}_S(p)\leq1+\sum_{j\in S}\mathbf{I}^{S\setminus\{j\}\cup[j-1]}_{S\cup[j-1]}(p)\leq1+\sum_{j\in S}C_{S\setminus\{j\}\cup[j-1]}=C_S\]
where the first inequality is due to supermodularity (Lemma \ref{lem_supermodularity}), and the second is by the induction hypothesis, since the smallest integer not in $S\setminus\{j\}\cup[j-1]$ is $j$.
\end{proof}

Our final ingredients towards the proof of Theorem \ref{thm_main_ub} are the following two lemmas relating the quantities from Definition \ref{def_C}.

\begin{lemma}\label{lem_order}
For finite subsets $S$, $T$ of $\mathbb{N}$, if $S$ precedes $T$ in the colex order, then $C_S<C_T$. In particular for any $n$, $C_{[n]\setminus\{n\}}<C_{[n]\setminus\{n-1\}}<\cdots<C_{[n]\setminus\{2\}}<C_{[n]\setminus\{1\}}$. 
\end{lemma}

\begin{proof}
It is sufficient to prove $C_S<C_T$ when $S$ is the set immediately preceding $T$ in the colex order, that is, $S=T\setminus\{i\}\cup[i-1]$, where $i$ is the smallest integer in $T$. But by Definition \ref{def_C}, we immediately have $C_S=C_{T\setminus\{i\}\cup[i-1]}<C_T$. Further for $i<j\leq n$, since $[n]\setminus\{j\}$ precedes $[n]\setminus\{i\}$ in the colex order, we have $C_{[n]\setminus\{j\}}<C_{[n]\setminus\{i\}}$.
\end{proof}

\begin{lemma}\label{lem_alpha}
For every $n\in\mathbb{N}$, $\alpha_n=\sum_{j=1}^nC_{[n]\setminus\{j\}}=C_{[n]}-1$.
\end{lemma}

\begin{proof}
We prove the claim by induction on $n$, noting that the claim holds for $n=0$. Let $n\geq1$ and assume $C_{[n-1]}=\alpha_{n-1}+1$. For any $S\subseteq[n]$ containing $n$, we first prove $C_S=(\alpha_{n-1}+2)C_{S\setminus\{n\}}$, by induction on the position of $S$ in the colex order.

In the colex order, the first subset of $[n]$ containing $n$ is $\{n\}$. For $S=\{n\}$, we have from Definition \ref{def_C} and by induction on $n$,
\[C_{\{n\}}=1+C_{[n-1]}=2+\alpha_{n-1}=(\alpha_{n-1}+2)C_{\emptyset}\]
For an arbitrary $S\subseteq[n]$ containing $n$, we have from Definition \ref{def_C} and by induction on $n$,
\[C_S=1+\sum_{j\in S}C_{S\setminus\{j\}\cup[j-1]}=1+C_{[n-1]}+\sum_{j\in S\setminus\{n\}}C_{S\setminus\{j\}\cup[j-1]}=\alpha_{n-1}+2+\sum_{j\in S\setminus\{n\}}C_{S\setminus\{j\}\cup[j-1]}\]
Since $S\setminus\{j\}\cup[j-1]$ precedes $S$ in the colex order, $C_{S\setminus\{j\}\cup[j-1]}=(\alpha_{n-1}+2)C_{S\setminus\{j,n\}\cup[j-1]}$. Thus,
\begin{eqnarray*}
C_S & = & \alpha_{n-1}+2+(\alpha_{n-1}+2)\sum_{j\in S\setminus\{n\}}C_{S\setminus\{j,n\}\cup[j-1]}\\
 & = & (\alpha_{n-1}+2)\left(1+\sum_{j\in S\setminus\{n\}}C_{S\setminus\{j,n\}\cup[j-1]}\right)=(\alpha_{n-1}+2)C_{S\setminus\{n\}}
\end{eqnarray*}

This proves $C_S=(\alpha_{n-1}+2)C_{S\setminus\{n\}}$, for all $S\subseteq[n]$ containing $n$. In particular, we have $C_{[n]}=(\alpha_{n-1}+2)C_{[n-1]}$. Again, by induction on $n$ and from Definition \ref{def_C}, we have
\[C_{[n]}=(\alpha_{n-1}+2)(\alpha_{n-1}+1)=\alpha_{n-1}^2+3\alpha_{n-1}+2=\alpha_n+1\]
\end{proof}

\subsection{Proof of Theorem \ref{thm_main_ub}}

We now show how the above lemmas imply Theorem \ref{thm_main_ub}. First, we prove an upper bound on the competitive ratio achieved by the probability distribution $p=(p_1,\ldots,p_k)$, when the weights are $\beta_1,\ldots,\beta_k$. 

\begin{theorem}\label{thm_competitive}
Consider an instance of the weighted $k$-server problem with weights $\beta_1,\ldots,\beta_k$, and a randomized memoryless algorithm which moves the $i^{\text{\tiny{th}}}$ server with a probability $p_i$, where $p_1\geq\cdots\geq p_k$. Then the competitive ratio of this algorithm against an adaptive online adversary is at most $\tilde{\alpha}(\beta,p)$, where
\[\tilde{\alpha}(\beta,p)=\left(\sum_{j=1}^kp_j\beta_j\right)\max_{i\in[k]}\frac{\mathbf{I}^{[k]\setminus\{i\}}_{[k]}(p)}{p_i\beta_i}\]
\end{theorem}

\begin{proof}
Lemma \ref{lem_eqn_varphi} assures that the constraints (\ref{eqn_varphi_alg}) hold. To satisfy the set of constraints given by (\ref{eqn_varphi_adv}), we choose
\[\alpha=\left(\sum_{j=1}^kp_j\beta_j\right)\max_{S\subseteq[k],i\in S}\frac{\varphi_S(p)-\varphi_{S\setminus\{i\}}(p)}{\beta_i}\]
Due to the supermodularity property from Lemma \ref{lem_supermodularity}, the maximum is attained for $S=[k]$. Thus we have
\[\alpha=\left(\sum_{j=1}^kp_j\beta_j\right)\max_{i\in[k]}\frac{\varphi_{[k]}(p)-\varphi_{[k]\setminus\{i\}}(p)}{\beta_i}=\left(\sum_{j=1}^kp_j\beta_j\right)\max_{i\in[k]}\frac{\mathbf{I}^{[k]\setminus\{i\}}_{[k]}(p)}{p_i\beta_i}\]
\end{proof}

With Theorem \ref{thm_competitive} in place, we are ready to prove Theorem \ref{thm_main_ub}.

\begin{proof}[\textbf{Proof of Theorem \ref{thm_main_ub}}]
Let $\beta_1,\ldots,\beta_k$ be the weights of the servers, and assume $\beta_1\leq\cdots\leq\beta_k$, without loss of generality. The required memoryless algorithm behaves as follows. Let $p_i=C_{[k]\setminus\{i\}}/\beta_i$ for all $i$. On receiving a request which is not covered by any server, the algorithm serves it with the $i^{\text{\tiny{th}}}$ server with probability $p_i/P$, where $P=\sum_{j=1}^kp_j$. By Lemma \ref{lem_order} and our assumption: $\beta_1\leq\cdots\leq\beta_k$, we have $p_1\geq\ldots\geq p_k$. Thus, we can apply Theorem \ref{thm_competitive}, and hence, the competitive ratio of our algorithm is at most
\[\tilde{\alpha}(\beta,p)=\left(\sum_{j=1}^kp_j\beta_j\right)\max_{i\in[k]}\frac{\mathbf{I}^{[k]\setminus\{i\}}_{[k]}(p)}{p_i\beta_i}=\left(\sum_{j=1}^kC_{[k]\setminus\{j\}}\right)\max_{i\in[k]}\frac{\mathbf{I}^{[k]\setminus\{i\}}_{[k]}(p)}{C_{[k]\setminus\{i\}}}\leq\alpha_k\]
where the last inequality follows from Lemma \ref{lem_alpha}, and Lemma \ref{lem_currents_abs_bound}. Note that since the currents are invariant under scaling of $p$, so is $\tilde{\alpha}(\beta,p)$, and hence, $P$ can be ignored.
\end{proof}

\begin{corollary}[to Theorem \ref{thm_competitive}]
The Harmonic algorithm for the weighted $k$-server problem on uniform spaces has a competitive ratio of $k\alpha_k$ against an online adaptive adversary.
\end{corollary}

\begin{proof}
The probabilities for the Harmonic algorithm are given by $p_i=(1/\beta_i)/\sum_{j=1}^k(1/\beta_j)$. Therefore, $1/(p_i\beta_i)=\sum_{j=1}^k(1/\beta_j)$ for all $i$. Also, $\sum_{j=1}^kp_j\beta_j=k/\sum_{j=1}^k(1/\beta_j)$. By Theorem \ref{thm_competitive} the competitive ratio is given by
\[\alpha=\left(\sum_{j=1}^kp_j\beta_j\right)\max_{i\in[k]}\frac{\mathbf{I}^{[k]\setminus\{i\}}_{[k]}(p)}{p_i\beta_i}=k\cdot\max_{i\in[k]}\mathbf{I}^{[k]\setminus\{i\}}_{[k]}(p)\leq k\cdot\max_{i\in[k]}C_{[k]\setminus\{i\}}\leq k\cdot\sum_{i=1}^kC_{[k]\setminus\{i\}}=k\alpha_k\]
\end{proof}

\section{Proof of the Lower Bound}\label{sec_lowerbound}

In this section, we show that it is not possible to improve the upper bound of $\alpha_k$ on the competitive ratio of randomized memoryless algorithms for the weighted $k$-server problem, on uniform spaces. We will exhibit costs $\beta$ such that, irrespective of the probability distribution chosen by an algorithm, an adversary can force a competitive ratio approaching $\alpha_k$.

\subsection{Constructing Adversaries}

As a first step towards proving Theorem \ref{thm_main_lb}, we prove that Theorem \ref{thm_competitive} is essentially tight. For an algorithm which uses probabilities $p=(p_1,\ldots,p_k)$ when the weights are $\beta=(\beta_1,\ldots,\beta_k)$, we prove a lower bound on the competitive ratio, which goes arbitrarily close to $\tilde{\alpha}(\beta,p)$, as the separation between the weights grows unbounded.

Fix $\beta_1\leq\cdots\leq\beta_k$, the weights of the servers, and let $s=\max_{1\leq i<k}\beta_i/\beta_{i+1}$. Fix some online algorithm. For each $t\in[k]$, we define an adversary $\mathcal{A}_t$, who gives requests from a uniform metric space with $2k+1$ points. As before, at any point of time let $a_i$ (resp. $s_i$) denote the position of the adversary's (resp. algorithm's) $i^{\text{\tiny{th}}}$ server. The adversary maintains the following invariant whenever it gives a request.
\begin{equation}\label{eqn_adv_invariant}
a_i\neq s_j\text{ for all }i<j\text{; }i,j\in[k]\text{ and }a_i\text{'s are all distinct.}
\end{equation}
The strategy of the adversary $\mathcal{A}_t$ is the following.
\begin{enumerate}
\item\label{item_req_eq} If $a_i=s_i$ for all $i\in[k]$, then move the $t^{\text{\tiny{th}}}$ server to a point not occupied by any of the $2k$ servers, (in other words different from $a_i$ and $s_i$ for all $i$), and request that point.
\item\label{item_req_neq} Else, find the smallest $i$ such that $a_i\neq s_i$. (Invariant (\ref{eqn_adv_invariant}) ensures that $a_i$ is not occupied by the algorithm.) Request $a_i$.
\item\label{item_adjust} If invariant (\ref{eqn_adv_invariant}) is violated for some $i,j$ after the algorithm serves the request, then move the $i^{\text{\tiny{th}}}$ server to a point not occupied by any of the $2k$ servers.
\end{enumerate}

Note that $t$ plays a role only in step \ref{item_req_eq}, and that the adversary pays only in steps \ref{item_req_eq} and \ref{item_adjust}. Let $ADV$ and $ADV'$ denote the total cost paid by the adversary in steps \ref{item_req_eq} and \ref{item_adjust} respectively, and let $ALG$ be the (expected) total cost paid by the online algorithm. The following lemma is immediate.

\begin{lemma}\label{lem_advprime}
$ALG\geq ADV'/s$.
\end{lemma}

\begin{proof}
Every time the adversary executes step \ref{item_adjust} and moves its $i^{\text{\tiny{th}}}$ server out of a point, the algorithm must have moved its $j^{\text{\tiny{th}}}$ server, for some $j>i$, to that point from elsewhere. Thus, the algorithm paid $\beta_j\geq\beta_i/s$, whereas the adversary pays $\beta_i$.
\end{proof}

Note that the above lemma holds even if the algorithm is not memoryless. Now the next two claims assume that the algorithm is memoryless, and prove lower bounds on its competitive ratio. Fix a randomized memoryless algorithm, which moves the $i^{\text{\tiny{th}}}$ server with probability $p_i$, whenever there is no server on the requested point.

\begin{theorem}\label{thm_lowerbound}
Let $t\in[k]$ be such that
\[\tilde{\alpha}(\beta,p)=\left(\sum_{j=1}^kp_j\beta_j\right)\max_{i\in[k]}\frac{\mathbf{I}^{[k]\setminus\{i\}}_{[k]}(p)}{p_i\beta_i}=\left(\sum_{j=1}^kp_j\beta_j\right)\cdot\frac{\mathbf{I}^{[k]\setminus\{t\}}_{[k]}(p)}{p_t\beta_t}\]
(in other words, $i=t$ achieves the maximum). Then the competitive ratio of the algorithm against $\mathcal{A}_t$ is at least $\tilde{\alpha}(\beta,p)/(1+s\tilde{\alpha}(\beta,p))$.
\end{theorem}

\begin{proof}
We prove $ALG\geq\tilde{\alpha}(\beta,p)ADV$. The theorem follows from this and Lemma \ref{lem_advprime}, since the total cost paid by the adversary is $ADV+ADV'$. As before, at any point of time let $S=\{i\suchthat a_i=s_i\}\subseteq[k]$, and $S$ will denote the state of the system. We will again assign a potential to each state, but this time we will ensure the following.
\begin{enumerate}
\item When the adversary $\mathcal{A}_t$ moves its $t^{\text{\tiny{th}}}$ server in step \ref{item_req_eq}, the increase in potential is \textbf{at least} $\tilde{\alpha}(\beta,p)\cdot\beta_t$.
\item When the algorithm is moves a server, the expected decrease in potential is \textbf{at most} the expected cost paid by the algorithm.
\end{enumerate}
Note that when the adversary moves its servers in step \ref{item_adjust} to ensure $a_i\neq s_j$ for all $i<j$, the state remains the same. Thus, the above two statements imply $ALG\geq\tilde{\alpha}(\beta,p)ADV$. Interestingly, the potentials that we assign to the states here are same as those that we assigned in the proof of the upper bound. That is, $\phi_S=-(\sum_{j=1}^kp_j\beta_j)\varphi_S(p)$. Note however, that we have not made any assumption about whether $p$ is a non-decreasing sequence.

Consider the situation when the adversary incurs a cost of $\beta_t$, in step \ref{item_req_eq}. Since $a_i=s_i$ for all $i$, the state is $[k]$. The state after the move is $[k]\setminus\{t\}$, and the change in potential is
\[\phi_{[k]\setminus\{t\}}-\phi_{[k]}=\left(\sum_{j=1}^kp_j\beta_j\right)(\varphi_{[k]}(p)-\varphi_{[k]\setminus\{t\}}(p))=\left(\sum_{j=1}^kp_j\beta_j\right)\cdot\frac{\mathbf{I}^{[k]\setminus\{t\}}_{[k]}(p)}{p_{t}}=\tilde{\alpha}(\beta,p)\cdot\beta_{t}\]

Now, consider the algorithm's move in response to a request, when the system is in state $S\subsetneq[k]$. Let $i$ be the smallest integer not in $S$. By step \ref{item_req_neq} of the adversary, the next request is $a_i$, and this point is not occupied by any of the algorithm's servers. Hence, the algorithm must incur a cost $\sum_{j=1}^kp_j\beta_j$ in expectation. The expected change in potential is
\begin{eqnarray*}
 & & p_i\left(\phi_{S\cup\{i\}}-\phi_S\right)+\sum_{j\in S}p_j\left(\phi_{S\setminus\{j\}}-\phi_S\right)\\
 & = & -\left(\sum_{j=1}^kp_j\beta_j\right)\left(p_i\left(\varphi_{S\cup\{i\}}(p)-\varphi_S(p)\right)-\sum_{j\in S}p_j\left(\varphi_S(p)-\varphi_{S\setminus\{j\}}(p)\right)\right)\\
 & = & -\sum_{j=1}^kp_j\beta_j
\end{eqnarray*}
where the last equality follows from Lemma \ref{lem_eqn_varphi_tight}. Thus, we have proved $ALG\geq\tilde{\alpha}(\beta,p)ADV$.
\end{proof}

Since $\beta_1\leq\cdots\leq\beta_k$, we expect a reasonable algorithm to choose probabilities $p_1\geq\cdots\geq p_k$, at least when the ratio $\beta_{i+1}/\beta_i$ is sufficiently large for all $i$. We prove the next lemma in order to rule out the possibility of a ``counter-intuitive'' algorithm being competitive, where the algorithm always chooses $p_j>p_i$ for some $j>i$, no matter how large $\beta_j/\beta_i$ is.

\begin{lemma}\label{lem_lowerbound}
For any $i\in[k]$, the competitive ratio of the algorithm against $\mathcal{A}_i$ is at least $\gamma_i(\beta,p)/(1+s\gamma_i(\beta,p))$, where $\gamma_i(\beta,p)=\left(\sum_{j=1}^kp_j\beta_j\right)/p_i\beta_i$.
\end{lemma}

\begin{proof}
Analogous to Theorem \ref{thm_lowerbound}, it is sufficient to prove $ALG\geq\gamma_i(\beta,p)ADV$. Say that a new phase begins whenever $a_j=s_j$ for all $j\in[k]$. We prove that in every phase the change in $ALG$ is at least $\gamma_i(\beta,p)$ times the change in $ADV$. Observe that the $i^{\text{\tiny{th}}}$ server of the algorithm must move at least once in every phase. Suppose this happens for the first time on the $m^{\text{\tiny{th}}}$ request. Then $m$ is a geometrically distributed random variable with parameter $p_i$ and hence $\mathbb{E}[m]=1/p_i$.

For each of the first $m-1$ requests, the algorithm does not move its $i^{\text{\tiny{th}}}$ server, and moves its $j^{\text{\tiny{th}}}$ server with probability $p_j/(1-p_i)$, for $j\neq i$. Hence the algorithm pays $\left(\sum_{j\in[k]\setminus\{i\}}p_j\beta_j\right)/(1-p_i)$ in expectation on each of the first $m-1$ requests, and $\beta_i$ on the $m^{\text{\tiny{th}}}$ one. The total expected cost paid on the first $m$ requests, conditioned on $m$, is
\[(m-1)\times\frac{\sum_{j\in[k]\setminus\{i\}}p_j\beta_j}{1-p_i}+\beta_i\]
Thus, the expected cost paid by the algorithm in a phase, until it moves its $i^{\text{\tiny{th}}}$ server for the first time, is given by
\begin{eqnarray*}
\mathbb{E}\left[(m-1)\times\frac{\sum_{j\in[k]\setminus\{i\}}p_j\beta_j}{1-p_i}+\beta_i\right] & = & \mathbb{E}[m-1]\times\frac{\sum_{j\in[k]\setminus\{i\}}p_j\beta_j}{1-p_i}+\beta_i\\
 & = & \left(\frac{1}{p_i}-1\right)\times\frac{\sum_{j\in[k]\setminus\{i\}}p_j\beta_j}{1-p_i}+\beta_i\\
 & = & \frac{\sum_{j\in[k]\setminus\{i\}}p_j\beta_j}{p_i}+\beta_i\\
 & = & \frac{\sum_{j\in[k]}p_j\beta_j}{p_i}=\gamma_i(\beta,p)\times\beta_i
\end{eqnarray*}
This is a lower bound on the change in $ALG$ in a phase. Further, step \ref{item_req_eq} of the adversary $\mathcal{A}_i$ is executed exactly once in a phase, and hence, $ADV$ increases by exactly $\beta_i$ in every phase. Thus, we have proved $ALG\geq\gamma_i(\beta,p)ADV$.
\end{proof}

\subsection{Proof of Theorem \ref{thm_main_lb}}

With Theorem \ref{thm_lowerbound} in place, proving Theorem \ref{thm_main_lb} reduces to proving that $\frac{\tilde{\alpha}(\beta,p)}{1+s\tilde{\alpha}(\beta,p)}$ can be forced to be arbitrarily close to $\alpha_k$, where $s=\max_{1\leq i<k}\beta_i/\beta_{i+1}$, with a suitably chosen $\beta$. Let the weights of the servers, parameterized by $r>1$, be given by $\beta_i(r)=r^{i-1}$. Consider a randomized memoryless algorithm with a bounded competitive ratio, which chooses a probability distribution $p(r)$ for the weights $\beta(r)$. If $\limsup_{r\rightarrow\infty}p_j(r)/p_i(r)=\delta>0$ for some $i<j$, then there exist arbitrarily large $r$ such that $p_j(r)/p_i(r)\geq\delta$. Then by Lemma \ref{lem_lowerbound}, the adversary $\mathcal{A}_i$ forces an unbounded lower bound on the competitive ratio, which is a contradiction. Thus, we must have $\limsup_{r\rightarrow\infty}p_j(r)/p_i(r)=0$, and since $p_j(r)/p_i(r)\geq0$, we have
\begin{equation}\label{eqn_lim_p}
\lim_{r\rightarrow\infty}\frac{p_j(r)}{p_i(r)}=0\text{ for all }i<j\text{; }i,j\in[k]
\end{equation} 
Thus, for a sufficiently large $r$, we must have ${p_1(r)\geq\cdots\geq p_k(r)}$. As a consequence of the next lemma, we prove that the supermodularity inequalities, that we applied in the proof of Lemma \ref{lem_currents_abs_bound}, are all tight in the limit as $r\rightarrow\infty$.

\begin{lemma}
Let $S\subsetneq[k]$ and $i$ be the smallest integer not in $S$. Suppose $j\notin S$ and $i<j<k$. Then $\lim_{r\rightarrow\infty}\left[\mathbf{I}^{S\cup\{i\}}_{S\cup\{i,j\}}(p(r))-\mathbf{I}^{S}_{S\cup\{j\}}(p(r))\right]=0$.
\end{lemma}

\begin{proof}
On one hand we have
\begin{eqnarray*}
\varphi_{S\cup\{i,j\}}(p(r))-\varphi_S(p(r)) & = & \left[\varphi_{S\cup\{i,j\}}(p(r))-\varphi_{S\cup\{i\}}(p(r))\right]+\left[\varphi_{S\cup\{i\}}(p(r))-\varphi_S(p(r))\right]\\
 & = & \frac{\mathbf{I}^{S\cup\{i\}}_{S\cup\{i,j\}}(p(r))}{p_j(r)}+\frac{\mathbf{I}^S_{S\cup\{i\}}(p(r))}{p_i(r)}
\end{eqnarray*}
On the other hand
\begin{eqnarray*}
\varphi_{S\cup\{i,j\}}(p(r))-\varphi_S(p(r)) & = & \left[\varphi_{S\cup\{i,j\}}(p(r))-\varphi_{S\cup\{j\}}(p(r))\right]+\left[\varphi_{S\cup\{j\}}(p(r))-\varphi_S(p(r))\right]\\
 & = & \frac{\mathbf{I}^{S\cup\{j\}}_{S\cup\{i,j\}}(p(r))}{p_i(r)}+\frac{\mathbf{I}^{S}_{S\cup\{j\}}(p(r))}{p_j(r)}
\end{eqnarray*}
Thus
\begin{eqnarray*}
\frac{\mathbf{I}^{S\cup\{i\}}_{S\cup\{i,j\}}(p(r))}{p_j(r)}+\frac{\mathbf{I}^{S}_{S\cup\{i\}}(p(r))}{p_i(r)} & = & \frac{\mathbf{I}^{S\cup\{j\}}_{S\cup\{i,j\}}(p(r))}{p_i(r)}+\frac{\mathbf{I}^{S}_{S\cup\{j\}}(p(r))}{p_j(r)}\\
\mathbf{I}^{S\cup\{i\}}_{S\cup\{i,j\}}(p(r))-\mathbf{I}^{S}_{S\cup\{j\}}(p(r)) & = & \frac{p_j(r)}{p_i(r)}\left[\mathbf{I}^{S\cup\{j\}}_{S\cup\{i,j\}}(p(r))-\mathbf{I}^{S}_{S\cup\{i\}}(p(r))\right]
\end{eqnarray*}
Taking limits as $r\rightarrow\infty$, noting that currents are bounded (Lemma \ref{lem_currents_abs_bound}), and using (\ref{eqn_lim_p}), we get the required result.
\end{proof}

By repeatedly applying the above lemma, we have for any $S\subsetneq[k]$ and $j\notin S$,
\[\lim_{r\rightarrow\infty}\left[\mathbf{I}^{S\cup[j-1]}_{S\cup[j]}(p(r))-\mathbf{I}^S_{S\cup\{j\}}(p(r))\right]=0\]
or in other words, for $S\subseteq[k]$ and $j\in S$
\begin{equation}\label{eqn_supermod_tight}
\lim_{r\rightarrow\infty}\left[\mathbf{I}^{S\setminus\{j\}\cup[j-1]}_{S\cup[j-1]}(p(r))-\mathbf{I}^{S\setminus\{j\}}_S(p(r))\right]=0
\end{equation}

Recall that in Lemma \ref{lem_currents_abs_bound}, we proved that if $i$ is the smallest integer not in $S$, then $\mathbf{I}^{S}_{S\cup\{i\}}(p)\leq C_S$ for any non-increasing $p$, where the constant $C_S$ was given by Definition \ref{def_C}. We will now prove that as $r$ goes to $\infty$, $\mathbf{I}^{S}_{S\cup\{i\}}(p(r))$ approaches $C_S$.

\begin{lemma}\label{lem_currents_abs_bound_tight}
For any $S\subsetneq[k]$, let $i$ be the smallest integer not in $S$. Then $\lim_{r\rightarrow\infty}\mathbf{I}^{S}_{S\cup\{i\}}(p(r))=C_S$.
\end{lemma}

\begin{proof}
We again prove the statement by induction on the position of $S$ in the colex order given by Definition \ref{def_colex}. For the base case, when $S=\emptyset$ we indeed have $\mathbf{I}^{\emptyset}_{\{1\}}(p(r))=1=C_{\emptyset}$ for all $r$.

For the inductive case, assume the claim holds for all finite subsets of $\mathbb{N}$ which precede a set $S$ in the colex order. Let $i$ be the smallest integer not in $S$. Then we have
\begin{eqnarray*}
\lim_{r\rightarrow\infty}\mathbf{I}^S_{S\cup\{i\}}(p(r)) & = & 1+\sum_{j\in S}\lim_{r\rightarrow\infty}\mathbf{I}^{S\setminus\{j\}}_S(p(r))=1+\sum_{j\in S}\lim_{r\rightarrow\infty}\mathbf{I}^{S\setminus\{j\}\cup[j-1]}_{S\cup[j-1]}(p(r))\\
 & = & 1+\sum_{j\in S}C_{S\setminus\{j\}\cup[j-1]}=C_S
\end{eqnarray*}
where the first equality is by Lemma \ref{lem_eqn_varphi_tight}, second due to (\ref{eqn_supermod_tight}), third by induction hypothesis, since the smallest integer not in $S\setminus\{j\}\cup[j-1]$ is $j$, and the fourth by Definition \ref{def_C}.
\end{proof}

\begin{proof}[\textbf{Proof of Theorem \ref{thm_main_lb}}]
We have
\begin{eqnarray}
\liminf_{r\rightarrow\infty}\tilde{\alpha}(\beta(r),p(r)) & = & \liminf_{r\rightarrow\infty}\left(\sum_{j=1}^kp_j(r)\beta_j(r)\right)\max_{i\in[k]}\frac{\mathbf{I}^{[k]\setminus\{i\}}_{[k]}(p(r))}{p_i(r)\beta_i(r)}\nonumber\\
 & \geq & \liminf_{r\rightarrow\infty}\left(\sum_{j=1}^kp_j(r)\beta_j(r)\right)\times\frac{\sum_{i=1}^k\mathbf{I}^{[k]\setminus\{i\}}_{[k]}(p(r))}{\sum_{i=1}^kp_i(r)\beta_i(r)}\nonumber\\
 & = & \liminf_{r\rightarrow\infty}\sum_{i=1}^k\mathbf{I}^{[k]\setminus\{i\}}_{[k]}(p(r))=\sum_{i=1}^kC_{[k]\setminus\{i\}}=\alpha_k\label{eqn_liminf}
\end{eqnarray}
where the penultimate equality is given by Lemma \ref{lem_currents_abs_bound_tight}, and the last one by Lemma \ref{lem_alpha}. Thus,
\begin{eqnarray*}
\liminf_{r\rightarrow\infty}\frac{\tilde{\alpha}(\beta(r),p(r))}{1+\tilde{\alpha}(\beta(r),p(r))/r} & = & \liminf_{r\rightarrow\infty}\left(\frac{1}{\tilde{\alpha}(\beta(r),p(r))}+\frac{1}{r}\right)^{-1}=\left[\limsup_{r\rightarrow\infty}\left(\frac{1}{\tilde{\alpha}(\beta(r),p(r))}+\frac{1}{r}\right)\right]^{-1}\\
 & \geq & \left[\limsup_{r\rightarrow\infty}\frac{1}{\tilde{\alpha}(\beta(r),p(r))}+\limsup_{r\rightarrow\infty}\frac{1}{r}\right]^{-1}
\geq\alpha_k
\end{eqnarray*}
where the first inequality follows from sub-additivity of the $\limsup$ operator, and the last inequality from (\ref{eqn_liminf}). Thus, for any $\varepsilon>0$, there exists an $R$ such that for all $r>R$, we have
\[\frac{\tilde{\alpha}(\beta(r),p(r))}{1+\tilde{\alpha}(\beta(r),p(r))/r}\geq\alpha_k-\varepsilon\]
Using Theorem \ref{thm_lowerbound} with $s=\max_{1\leq i<k}\beta_i(r)/\beta_{i+1}(r)=1/r$, we conclude that the competitive ratio of the algorithm is no less than $\alpha_k$.
\end{proof}

\section{Concluding Remarks}

We have proved that there exists a competitive memoryless algorithm for the weighted $k$-server problem on uniform metric spaces. This is in contrast to the line metric, which does not admit a competitive memoryless algorithm, even with two servers. The competitive ratio $\alpha_k$, that we establish, is given by $\alpha_k=\alpha_{k-1}^2+3\alpha_{k-1}+1$. We can bound $\alpha_k$ as follows. We have $\alpha_k+2=(\alpha_{k-1}+2)^2-(\alpha_{k-1}+2)+1<(\alpha_{k-1}+2)^2$. Therefore, $\alpha_k+2<(\alpha_t+2)^{2^{k-t}}=[(\alpha_t+2)^{2^{-t}}]^{2^k}$ for any $t<k$. For $t=4$, one can verify that $(\alpha_t+2)^{2^{-t}}<1.6$, and hence $\alpha_k<1.6^{2^k}$, as promised in the introduction. We have also proved that $\alpha_k$ is the best possible competitive ratio of memoryless algorithms for the weighted $k$-server problem on uniform metric spaces. With this, we settle the problem completely.

The immediate increment to our results would perhaps be to determine whether there exists a competitive memoryless algorithm for the weighted server problem on star metrics. This problem translates to having a weight $\beta_i$ for the $i^{\text{\tiny{th}}}$ cache location, and a cost $c_t$ with each page $t$; the overall cost of replacing page $t$ by page $t'$ at the $i^{\text{\tiny{th}}}$ cache location being $\beta_i(c_t+c_{t'})$. It would be interesting to see whether our techniques work on the star metric too.

We improve the upper bound on the deterministic competitive ratio by \cite{FiatR94} for the weighted server problem on uniform metrics. However, our bound is still doubly exponential, whereas the lower bound is only exponential in the number of servers. It would be interesting to reduce this gap. The prime candidate for improving the upper bound is perhaps the (generalized) work function algorithm, which has been proved to be optimally competitive for the weighted $2$-server problem on uniform metrics \cite{ChrobakS04}, and which is the best known algorithm for several other problems \cite{KoutsoupiasP95, Burley96}.

The introduction of different costs for replacements at different cache locations seems to make the caching problem notoriously hard. This is certified by the fact that attempts to develop algorithms better than the one by Fiat and Ricklin \cite{FiatR94} have given negligible success even with $k=3$. For $k=2$, Sitters \cite{Sitters14} has shown that the generalized work function algorithm is competitive for the generalized server problem on arbitrary metrics, which subsumes the weighted $2$-server problem. He has also expressed a possibility of the non-existence of a competitive algorithm for $k>2$. All this is in a striking contrast with the problem of weighted caching, where the pages (points) have costs instead of cache locations (servers). For the weighted caching problem $k$-competitive deterministic and $O(\log k)$-competitive randomized algorithms have been discovered \cite{chrobak_sundar_SODA90, Manasse_JOA90, Young94, Young98, BansalBN12, AdamaszekCER12}, even when the pages have different sizes, matching the respective lower bounds.

\section*{Acknowledgment}
The authors would like to thank Nikhil Bansal for pointing them to some references.

\bibliographystyle{plain}
\bibliography{../../../mssms/references.bib,../../references.bib,../../../caching/references.bib}

\section*{Appendix}
\appendix

\section{Proof of Lemma \ref{lemma_f}}\label{sec_proof_lemma_f}

Throughout this section, we assume $p=(p_1,\ldots,p_n)$ is such that $p_1\geq\cdots\geq p_n>0$. In order to prove Lemma \ref{lemma_f}, we use the Gauss-Seidel trick on the system given by equations (\ref{eqn_f1}) and (\ref{eqn_f2}). In every iteration, we calculate an approximation to $f_S(p)$ in decreasing order of $|S|$. For $S\subseteq[n]$ containing $n$, we take $f^0_S(p)=0$ if $S\neq[n]$ and $f^t_{[n]}(p)=1+\sum_{j=1}^{n-1}\mathbf{I}^{[n-1]\setminus\{j\}}_{[n-1]}(p)$ for all $t$. Having obtained $f^{t-1}_{S'}(p)$ for each such $S'$, and $f^t_{S'}(p)$ for each such $S'\supseteq S\neq[n]$, we obtain $f^t_S(p)$ using the following update rule.
\begin{equation}\label{eqn_iter}
\left(p_i+\sum_{j\in S}p_j\right)f^t_S(p)=p_if^t_{S\cup\{i\}}(p)+\sum_{j\in S\setminus\{n\}}p_jf^{t-1}_{S\setminus\{j\}}(p)
\end{equation}
The system given by equations (\ref{eqn_f1}) and (\ref{eqn_f2}) becomes strictly diagonally dominant under the substitution $X_j=p_j$, since each $p_j>0$. Therefore, the approximations converge to the solution of the system. That is, $\lim_{t\to\infty}f^t_S(p)=f_S(p)$. We prove some claims about these iterated solutions.

\begin{claim}[Iteration monotonicity]\label{claim_f_nondec}
$f^t_S(p)$ is non-decreasing with respect to $t$, in other words, $f^{t-1}_S(p)\leq f^t_S(p)$ for every $S$.
\end{claim}

\begin{proof}
We prove by induction on $t$, and reverse induction on $|S|$. The claim is obvious for $t=1$, since $f^t_{[n]}(p)\geq 0$ does not change with $t$, and for any other $S$, $f^0_S(p)=0$ and $f^1_S(p)\geq 0$. Now, assuming the claim for $t-1$, and for all $S'\supseteq S$ in the current ($t^{\text{\tiny{th}}}$) iteration, we have for $S\neq[n]$
\begin{eqnarray*}
\left(p_i+\sum_{j\in S}p_j\right)f^t_S(p) & = & p_if^t_{S\cup\{i\}}(p)+\sum_{j\in S\setminus\{n\}}p_jf^{t-1}_{S\setminus\{j\}}(p)\\
 & \geq & p_if^{t-1}_{S\cup\{i\}}(p)+\sum_{j\in S\setminus\{n\}}p_jf^{t-2}_{S\setminus\{j\}}(p)\\
 & = & \left(p_i+\sum_{j\in S}p_j\right)f^{t-1}_S(p)
\end{eqnarray*}
\end{proof}

We now derive an equation which will be used repeatedly in the next claim. Let $i$ be the smallest integer not in $S$, and let $i'<n$, $i'\neq i$ also not be in $S$. Then the smallest integer not in $S\cup\{i'\}$ is $i$. Therefore, by equation (\ref{eqn_iter}) applied to $S\cup\{i'\}$, we have
\begin{equation}\label{eqn_common1}
\left(p_i+p_{i'}+\sum_{j\in S}p_j\right)f^t_{S\cup\{i'\}}(p)=p_if^t_{S\cup\{i,i'\}}(p)+p_{i'}f^{t-1}_S(p)+\sum_{j\in S\setminus\{n\}}p_jf^{t-1}_{S\cup\{i'\}\setminus\{j\}}(p)
\end{equation}
Adding $p_{i'}f^t_S(p)$ to both sides of equation (\ref{eqn_iter}), we get
\begin{equation}\label{eqn_common2}
\left(p_i+p_{i'}+\sum_{j\in S}p_j\right)f^t_S(p)=p_if^t_{S\cup\{i\}}(p)+p_{i'}f^t_S(p)+\sum_{j\in S\setminus\{n\}}p_jf^{t-1}_{S\setminus\{j\}}(p)
\end{equation}
Subtracting (\ref{eqn_common2}) from (\ref{eqn_common1}), we get
\begin{eqnarray}\label{eqn_common_expr}
\left(p_i+p_{i'}+\sum_{j\in S}p_j\right)\left(f^t_{S\cup\{i'\}}(p)-f^t_S(p)\right) & = & p_i\left(f^t_{S\cup\{i,i'\}}(p)-f^t_{S\cup\{i\}}(p)\right)-p_{i'}\left(f^t_S(p)-f^{t-1}_S(p)\right)\nonumber\\
 & & +\sum_{j\in S\setminus\{n\}}p_j\left(f^{t-1}_{S\cup\{i'\}\setminus\{j\}}(p)-f^{t-1}_{S\setminus\{j\}}(p)\right)
\end{eqnarray}

The following claim states a version of Lemma \ref{lemma_f} for the iterated solutions. Lemma \ref{lemma_f} follows easily from this, taking limit as $t\rightarrow\infty$.

\begin{claim}\label{claim_f_t}
Let $S=[n]\setminus\{l_1,\ldots,l_m\}$ be such that $n\in S$, where $l_1<\cdots<l_m<n$. Then the following are true.
\begin{enumerate}
\item $f^t_{S\cup\{l_i\}}(p)\geq f^t_S(p)$ for $i=1,\ldots,m$
\item $p_{l_i}\left(f^t_{S\cup\{l_i\}}(p)-f^t_S(p)\right)\leq p_{l_{i+1}}\left(f^t_{S\cup\{l_{i+1}\}}(p)-f^t_S(p)\right)$ for $i=1,\ldots,m-1$
\end{enumerate}
\end{claim}

\begin{proof}
We prove this claim by induction on $t$. For $t=0$, the claim is obvious. So suppose $t>0$. We first consider the case when $i=1$. 

Consider Part 1 of the claim. The smallest integer not in $S$ is $l_1$. Rewriting (\ref{eqn_iter}), we have
\[p_{l_1}\left(f^t_{S\cup\{l_1\}}(p)-f^t_S(p)\right)=p_nf^t_S(p)+\sum_{j\in S\setminus\{n\}}p_j\left(f^t_S(p)-f^{t-1}_{S\setminus\{j\}}(p)\right)\]
By Claim \ref{claim_f_nondec} and induction on $t$, we have $f^t_S(p)\geq0$ and $f^t_S(p)-f^{t-1}_{S\setminus\{j\}}(p)\geq f^{t-1}_S(p)-f^{t-1}_{S\setminus\{j\}}(p)\geq0$. Hence $f^t_{S\cup\{l_1\}}(p)-f^t_S(p)\geq0$, and thus Part 1 is proved for $i=1$, for any $S$.

Consider Part 2 of the claim. When $|S|=n-1$, there is nothing to prove in Part 2. So assume $|S|<n-1$. $l_1$ and $l_2$ are respectively the smallest and second smallest integers not in $S$, and thus, $l_2$ is the smallest integer not in $S\cup\{l_1\}$. Therefore, we have
\begin{eqnarray*}
\left(p_{l_1}+\sum_{j\in S}p_j\right)f^t_S(p) & = & p_{l_1}f^t_{S\cup\{l_1\}}(p)+\sum_{j\in S\setminus\{n\}}p_jf^{t-1}_{S\setminus\{j\}}(p)\\
\left(p_{l_2}+p_{l_1}+\sum_{j\in S}p_j\right)f^t_{S\cup\{l_1\}}(p) & = & p_{l_2}f^t_{S\cup\{l_1,l_2\}}(p)+p_{l_1}f^{t-1}_S(p)+\sum_{j\in S\setminus\{n\}}p_jf^{t-1}_{S\cup\{l_1\}\setminus\{j\}}(p)
\end{eqnarray*}
Thus,
\begin{eqnarray}
\left(p_{l_2}+p_{l_1}+\sum_{j\in S}p_j\right)p_{l_1}\left(f^t_{S\cup\{l_1\}}(p)-f^t_S(p)\right) & = & p_{l_2}p_{l_1}f^t_{S\cup\{l_1,l_2\}}(p)-p_{l_1}^2f^t_{S\cup\{l_1\}}(p)\nonumber\\
 & & -p_{l_2}p_{l_1}f^t_S(p)+p_{l_1}^2f^{t-1}_S(p)\label{eqn_l_1}\\
 & & +\sum_{j\in S\setminus\{n\}}p_jp_{l_1}\left(f^{t-1}_{S\cup\{l_1\}\setminus\{j\}}(p)-f^{t-1}_{S\setminus\{j\}}(p)\right)\nonumber
\end{eqnarray}
Further, from equation (\ref{eqn_common_expr}),
\begin{eqnarray}
\left(p_{l_1}+p_{l_2}+\sum_{j\in S}p_j\right)p_{l_2}\left(f^t_{S\cup\{l_2\}}(p)-f^t_S(p)\right) & = & p_{l_1}p_{l_2}\left(f^t_{S\cup\{l_2,l_1\}}(p)-f^t_{S\cup\{l_1\}}(p)\right)\nonumber\\
 & & -p_{l_2}^2\left(f^t_S(p)-f^{t-1}_S(p)\right)\label{eqn_l_2}\\
 & & +\sum_{j\in S\setminus\{n\}}p_jp_{l_2}\left(f^{t-1}_{S\cup\{l_2\}\setminus\{j\}}(p)-f^{t-1}_{S\setminus\{j\}}(p)\right)\nonumber
\end{eqnarray}
We need to prove (\ref{eqn_l_1}) is at most (\ref{eqn_l_2}). By induction on $t$, for each $j\in S\setminus\{n\}$, we have
\[p_jp_{l_1}\left(f^{t-1}_{S\cup\{l_1\}\setminus\{j\}}(p)-f^{t-1}_{S\setminus\{j\}}(p)\right)\leq p_jp_{l_2}\left(f^{t-1}_{S\cup\{l_2\}\setminus\{j\}}(p)-f^{t-1}_{S\setminus\{j\}}(p)\right)\]
Canceling $p_{l_2}p_{l_1}f^t_{S\cup\{l_1,l_2\}}(p)$, we are left to prove
\[-p_{l_1}^2f^t_{S\cup\{l_1\}}(p)-p_{l_2}p_{l_1}f^t_S(p)+p_{l_1}^2f^{t-1}_S(p)\leq-p_{l_2}p_{l_1}f^t_{S\cup\{l_1\}}(p)-p_{l_2}^2\left(f^t_S(p)-f^{t-1}_S(p)\right)\]
Using $f^t_S(p)-f^{t-1}_S(p)\geq0$ by Claim \ref{claim_f_nondec}, and the fact that $p_{l_1}\geq p_{l_2}>0$, we have
\[-p_{l_2}p_{l_1}\left(f^t_S(p)-f^{t-1}_S(p)\right)\leq-p_{l_2}^2\left(f^t_S(p)-f^{t-1}_S(p)\right)\]
Therefore, it is sufficient to prove
\[-p_{l_1}^2f^t_{S\cup\{l_1\}}(p)-p_{l_2}p_{l_1}f^{t-1}_S(p)+p_{l_1}^2f^{t-1}_S(p)\leq-p_{l_2}p_{l_1}f^t_{S\cup\{i\}}(p)\]
that is,
\[p_{l_2}p_{l_1}f^t_{S\cup\{l_1\}}(p)-p_{l_2}p_{l_1}f^{t-1}_S(p)\leq p_{l_1}^2f^t_{S\cup\{l_1\}}(p)-p_{l_1}^2f^{t-1}_S(p)\]
that is,
\[p_{l_2}\left(f^t_{S\cup\{l_1\}}(p)-f^{t-1}_S(p)\right)\leq p_{l_1}\left(f^t_{S\cup\{l_1\}}(p)-f^{t-1}_S(p)\right)\]
This is true, since by Claim \ref{claim_f_nondec} and induction on $t$, we have $f^t_{S\cup\{l_1\}}(p)-f^{t-1}_S(p)\geq f^{t-1}_{S\cup\{l_1\}}(p)-f^{t-1}_S(p)\geq0$, and $p_{l_2}\leq p_{l_1}$.

We are now ready to prove the claim completely for any $S$ and $i$. We perform reverse induction on $|S|$, and then on $i$. We have proved the claim when $|S|=n-1$. So assume that $|S|<n-1$, and that the claim holds for all supersets of $S$. We have proved the claim for $i=1$, so assume $i>1$, and that the claim holds for $i-1$.

For Part 1, by induction on $i$, we have $f^t_{S\cup\{l_{i-1}\}}(p)-f^t_S(p)\geq0$ and $p_{l_{i-1}}(f^t_{S\cup\{l_{i-1}\}}(p)-f^t_S(p))\leq p_{l_i}(f^t_{S\cup\{l_i\}}(p)-f^t_S(p))$. Therefore $f^t_{S\cup\{l_i\}}(p)\geq f^t_S(p)$.

We now prove Part 2 for $i<m$. Since $f^t_{S\cup\{l_i\}}(p)\geq f^t_S(p)$ and $p_{l_i}\geq p_{l_{i+1}}$, it is sufficient to prove
\[\left(p_{l_1}+p_{l_i}+\sum_{j\in S}p_j\right)p_{l_i}\left(f^t_{S\cup\{l_i\}}(p)-f^t_S(p)\right)\leq\left(p_{l_1}+p_{l_{i+1}}+\sum_{j\in S}p_j\right)p_{l_{i+1}}\left(f^t_{S\cup\{l_{i+1}\}}(p)-f^t_S(p)\right)\]
Substituting from equation (\ref{eqn_common_expr}), we are left to prove
\[p_{l_1}p_{l_i}\left(f^t_{S\cup\{l_i,l_1\}}(p)-f^t_{S\cup\{l_1\}}(p)\right)-p_{l_i}^2\left(f^t_S(p)-f^{t-1}_S(p)\right)+\sum_{j\in S\setminus\{n\}}p_jp_{l_i}\left(f^{t-1}_{S\cup\{l_i\}\setminus\{j\}}(p)-f^{t-1}_{S\setminus\{j\}}(p)\right)\]
\[\leq\]
\[p_{l_1}p_{l_{i+1}}\left(f^t_{S\cup\{l_{i+1},l_1\}}(p)-f^t_{S\cup\{l_1\}}(p)\right)-p_{l_{i+1}}^2\left(f^t_S(p)-f^{t-1}_S(p)\right)\]
\[+\sum_{j\in S\setminus\{n\}}p_jp_{l_{i+1}}\left(f^{t-1}_{S\cup\{l_{i+1}\}\setminus\{j\}}(p)-f^{t-1}_{S\setminus\{j\}}(p)\right)\]
But by reverse induction on $|S|$, we have
\[p_{l_i}\left(f^t_{S\cup\{l_i,l_1\}}(p)-f^t_{S\cup\{l_1\}}(p)\right)\leq p_{l_{i+1}}\left(f^t_{S\cup\{l_{i+1},l_1\}}(p)-f^t_{S\cup\{l_1\}}(p)\right)\]
and by induction on $t$, for each $j\in S\setminus\{n\}$, we have
\[p_jp_{l_i}\left(f^{t-1}_{S\cup\{l_i\}\setminus\{j\}}(p)-f^{t-1}_{S\setminus\{j\}}(p)\right)\leq p_jp_{l_{i+1}}\left(f^{t-1}_{S\cup\{l_{i+1}\}\setminus\{j\}}(p)-f^{t-1}_{S\setminus\{j\}}(p)\right)\]
Further since $f^t_S(p)-f^{t-1}_S(p)\geq0$ by Claim \ref{claim_f_nondec}, and $p_{l_i}\geq p_{l_{i+1}}>0$, the claim stands proved.
\end{proof}

\begin{proof}[\textbf{Proof of Lemma \ref{lemma_f}}]
Let $S=[n]\setminus\{l_1,\ldots,l_m\}$, where $l_1<\cdots<l_m<n$, be a set containing $n$. Suppose $p=(p_1,\ldots,p_n)$ with $p_1\geq\cdots\geq p_n>0$. Then for all $i,i'\in\{l_1,\ldots,l_m\}$, $i<i'$, we are required to prove
\[p_i\left(f_{S\cup\{i\}}(p)-f_S(p)\right)\leq p_{i'}\left(f_{S\cup\{i'\}}(p)-f_S(p)\right)\]
Note that it is sufficient to prove the claim for $i=l_j$ and $i'=l_{j+1}$ for each $j$. From Part 2 of Claim \ref{claim_f_t}, we have for every $t$, $p_{l_j}(f^t_{S\cup\{l_j\}}(p)-f^t_S(p))\leq p_{l_{j+1}}(f^t_{S\cup\{l_{j+1}\}}(p)-f^t_S(p))$. Taking limit as $t\rightarrow\infty$, we get $p_{l_j}(f_{S\cup\{l_j\}}(p)-f_S(p))\leq p_{l_{j+1}}(f_{S\cup\{l_{j+1}\}}(p)-f_S(p))$.
\end{proof}

Similarly, using Part 1 of Claim \ref{claim_f_t} and taking limit as $t\rightarrow\infty$, we get $f_{S\cup\{i\}}(p)\geq f_S(p)$ if $i\notin S$. Using this fact and (\ref{eqn_current_f}), it is easy to prove inductively that as long as $p_1\geq\cdots\geq p_n>0$,
\begin{equation}\label{eqn_f_monotone}
\mathbf{I}^S_{S\cup\{i\}}(p)\geq0\text{ if }i\notin S
\end{equation}
In other words, all currents are non-negative.

\section{Proof of Lemma \ref{lem_currents}}

Assume $p=(p_1,\ldots,p_n)$ is such that $p_1\geq\cdots\geq p_n>0$. In order to prove Lemma \ref{lem_currents}, we require the following two claims, apart from Lemma \ref{lemma_f}.

\begin{claim}[Symmetry]\label{claim_symm}
For any $n$, let the sequence $p$ be such that $p_n=p_{n-1}$. Then we have for any $S\subseteq[n-1]$ such that $n-1\in S$
\begin{enumerate}
\item $f_S(p)=f_{S\setminus\{n-1\}\cup\{n\}}(p)$.
\item $\varphi_S(p)=\varphi_{S\setminus\{n-1\}\cup\{n\}}(p)$
\end{enumerate}
\end{claim}

\begin{proof}
Consider the subset of quantities $\{f_S(p)\suchthat n\in S\text{, }n-1\notin S\}$. These form a unique solution to the system given by
\begin{equation}\label{eqn_f1'}
\left(p_i+\sum_{j\in S}p_j\right)f_S(p)=p_if_{S\cup\{i\}}(p)+\sum_{j\in S\setminus\{n\}}p_jf_{S\setminus\{j\}}(p)
\end{equation}
for $S\neq[n-2]\cup\{n\}$, where $i<n-1$ is the smallest integer not in $S$, and
\begin{equation}\label{eqn_f2'}
\left(\sum_{j=1}^np_j\right)f_{[n-2]\cup\{n\}}(p)=p_{n-1}\left(1+\sum_{j=1}^{n-1}\mathbf{I}^{[n-1]\setminus\{j\}}_{[n-1]}(p)\right)+\sum_{j=1}^{n-2}p_jf_{[n-2]\cup\{n\}\setminus\{j\}}(p)
\end{equation}
since $1+\sum_{j=1}^{n-1}\mathbf{I}^{[n-1]\setminus\{j\}}_{[n-1]}(p)=f_{[n]}(p)$. Suppose we replace $f_S(p)$ by $f_{S\cup\{n-1\}\setminus\{n\}}(p)$ for each $S$. We will prove that equations (\ref{eqn_f1'}) and (\ref{eqn_f2'}) are still satisfied. This will imply the first part of the claim. Equation (\ref{eqn_f1'}) can be verified easily using the fact that $p_n=p_{n-1}$, and we are left to check (\ref{eqn_f2'}). That is, we have to prove
\[\left(\sum_{j=1}^np_j\right)f_{[n-1]}(p)=p_{n-1}\left(1+\sum_{j=1}^{n-1}\mathbf{I}^{[n-1]\setminus\{j\}}_{[n-1]}(p)\right)+\sum_{j=1}^{n-2}p_jf_{[n-1]\setminus\{j\}}(p)\]
Since $f_{[n-1]}(p)=\mathbf{I}^{[n-2]}_{[n-1]}(p)$ and $p_n=p_{n-1}$, we are left to prove
\[\left(\sum_{j=1}^{n-1}p_j\right)f_{[n-1]}(p)=p_{n-1}\left(1+\sum_{j=1}^{n-2}\mathbf{I}^{[n-1]\setminus\{j\}}_{[n-1]}(p)\right)+\sum_{j=1}^{n-2}p_jf_{[n-1]\setminus\{j\}}(p)\]
that is,
\[\left(\sum_{j=1}^{n-1}\frac{p_j}{p_{n-1}}\right)f_{[n-1]}(p)-\sum_{j=1}^{n-2}\frac{p_jf_{[n-1]\setminus\{j\}}(p)}{p_{n-1}}=1+\sum_{j=1}^{n-2}\mathbf{I}^{[n-1]\setminus\{j\}}_{[n-1]}(p)\]
that is,
\begin{equation}\label{eqn_claim_symm}
f_{[n-1]}(p)+\sum_{j=1}^{n-2}\frac{p_j(f_{[n-1]}(p)-f_{[n-1]\setminus\{j\}}(p))}{p_{n-1}}=1+\sum_{j=1}^{n-2}\mathbf{I}^{[n-1]\setminus\{j\}}_{[n-1]}(p)
\end{equation}
Using the fact from equation (\ref{eqn_current_f}), that
\[\mathbf{I}^{[n-1]\setminus\{j\}}_{[n-1]}(p)=\mathbf{I}^{[n-2]\setminus\{j\}}_{[n-2]}(p)+\frac{p_j(f_{[n-1]}(p)-f_{[n-1]\setminus\{j\}}(p))}{p_{n-1}}\]
and substituting this in the right hand side of (\ref{eqn_claim_symm}), we are left with
\[f_{[n-1]}(p)=1+\sum_{j=1}^{n-2}\mathbf{I}^{[n-2]\setminus\{j\}}_{[n-2]}(p)\]
which is indeed true. The second part of the claim follows from the first part and (\ref{eqn_def_varphi}).
\end{proof}

The next claim states that the function $f_S(p)$ increases, as $p_n$ decreases from $p_{n-1}$ to $0$, for fixed $p_1,\ldots,p_{n-1}$.

\begin{claim}[$p$-monotonicity]\label{claim_monotone_f}
Let $p'$ be such that $p'_j=p_j$ for $1\leq j<n$ and $p'_n\leq p_n$. Then for each $S\subseteq[n]$, $f_S(p')\geq f_S(p)$.
\end{claim}

\begin{proof}
When $n\notin S$, and also when $S=[n]$, the claim is obvious, since in this case, $f_S(p)$ depends only on $p_1,\ldots,p_{n-1}$. Else if $n\in S\subsetneq[n]$, then we again prove that the claim holds after each iteration of the Gauss-Seidel procedure described in Appendix \ref{sec_proof_lemma_f}, implying that the claim holds for the limit. Thus, we need to prove $f^t_S(p')\geq f^t_S(p)$ for each $t$. We again use induction on $t$, and reverse induction on $|S|$. For $t=0$ the claim is trivial. For $t>0$ we have proved the claim for $S=[n]$. So let $S\subsetneq[n]$, $n\in S$ and let $i<n$ be the smallest integer not in $S$. Assuming the claim for $t-1$ and for all $S'\supseteq S$ in the current $t^{\text{\tiny{th}}}$ iteration, we have
\begin{eqnarray*}
\left(p'_i+\sum_{j\in S}p'_j\right)f^t_S(p') & = & p'_if^t_{S\cup\{i\}}(p')+\sum_{j\in S\setminus\{n\}}p'_jf^{t-1}_{S\setminus\{j\}}(p')\\
 & \geq & p_if^t_{S\cup\{i\}}(p)+\sum_{j\in S\setminus\{n\}}p_jf^{t-1}_{S\setminus\{j\}}(p)\\
 & = & \left(p_i+\sum_{j\in S}p_j\right)f^t_S(p)
\end{eqnarray*}
Note that the inequality holds because $p'_i=p_i$ and $p'_j=p_j$ for all $j\in S\setminus\{n\}$. But we also have $p'_i+\sum_{j\in S}p'_j\leq p_i+\sum_{j\in S}p_j$, and therefore, $f^t_S(p')\geq f^t_S(p)$.
\end{proof}

\begin{proof}[\textbf{Proof of Lemma \ref{lem_currents}}]
Let $S$ be a finite subset of $\mathbb{N}$, $i,j\notin S$, $i<j$ and $S\cup\{i,j\}\subseteq[n]$. Suppose $p=(p_1,\ldots,p_n)$ with $p_1\geq\cdots\geq p_n>0$. We are required to prove
\[\mathbf{I}^S_{S\cup\{i\}}(p)\leq\mathbf{I}^S_{S\cup\{j\}}(p)\]
Note that we may assume, without loss of generality, that the largest integer in $S\cup\{j\}$ is $n$. We prove the claim by induction on $n$. For $n=1$, there is nothing to prove. For $n>1$ we consider two cases: $n>j$ and $n=j$.

In the former case, the largest integer in $S$ as well as $S\cup\{i\}$ is $n$. From equation (\ref{eqn_current_f}), we need to prove
\[\mathbf{I}^{S\setminus\{n\}}_{S\cup\{i\}\setminus\{n\}}(p)+\frac{p_i}{p_n}\left(f_{S\cup\{i\}}(p)-f_S(p)\right)\leq\mathbf{I}^{S\setminus\{n\}}_{S\cup\{j\}\setminus\{n\}}(p)+\frac{p_j}{p_n}\left(f_{S\cup\{j\}}(p)-f_S(p)\right)\]
But by induction hypothesis, we have
\[\mathbf{I}^{S\setminus\{n\}}_{S\cup\{i\}\setminus\{n\}}(p)\leq\mathbf{I}^{S\setminus\{n\}}_{S\cup\{j\}\setminus\{n\}}(p)\]
and from Lemma \ref{lemma_f}, we have
\[\frac{p_i}{p_n}\left(f_{S\cup\{i\}}(p)-f_S(p)\right)\leq\frac{p_j}{p_n}\left(f_{S\cup\{j\}}(p)-f_S(p)\right)\]
Adding these inequalities, we get the desired result.

In the latter case, that is $n=j$, $S\cup\{i\}\subseteq[n-1]$. We have $\mathbf{I}^S_{S\cup\{j\}}(p)=\mathbf{I}^S_{S\cup\{n\}}(p)=f_{S\cup\{n\}}(p)$, and this increases with decreasing $p_n$, due to Claim \ref{claim_monotone_f}. On the other hand, $\mathbf{I}^S_{S\cup\{i\}}(p)$ is independent of $p_n$, since $S\subseteq S\cup\{i\}\subseteq[n-1]$. Therefore, it is sufficient to prove the claim assuming $p_n=p_{n-1}$.

We consider two sub-cases: $n-1\in S$ and $n-1\notin S$. First, suppose $n-1\in S$. Then by Claim \ref{claim_symm}, we have 
\[p_i\left(\varphi_{S\cup\{i\}}(p)-\varphi_S(p)\right)=p_i\left(\varphi_{S\cup\{i\}\setminus\{n-1\}\cup\{n\}}(p)-\varphi_{S\setminus\{n-1\}\cup\{n\}}(p)\right)\]
which means $\mathbf{I}^S_{S\cup\{i\}}(p)=\mathbf{I}^{S\setminus\{n-1\}\cup\{n\}}_{S\cup\{i\}\setminus\{n-1\}\cup\{n\}}(p)$, and
\[p_n\left(\varphi_{S\cup\{n\}}(p)-\varphi_S(p)\right)=p_i\left(\varphi_{S\cup\{n\}}(p)-\varphi_{S\setminus\{n-1\}\cup\{n\}}(p)\right)\]
which means $\mathbf{I}^S_{S\cup\{n\}}(p)=\mathbf{I}^{S\setminus\{n-1\}\cup\{n\}}_{S\cup\{n\}}(p)$. Since $i<n-1$, by the earlier case (taking $j=n-1$) we have $\mathbf{I}^{S\setminus\{n-1\}\cup\{n\}}_{S\cup\{i\}\setminus\{n-1\}\cup\{n\}}(p)\leq\mathbf{I}^{S\setminus\{n-1\}\cup\{n\}}_{S\cup\{n\}}(p)$, and hence $\mathbf{I}^S_{S\cup\{i\}}(p)\leq\mathbf{I}^S_{S\cup\{n\}}(p)$, as required.

Finally, suppose $n-1\notin S$. We already have $\mathbf{I}^S_{S\cup\{i\}}(p)\leq\mathbf{I}^S_{S\cup\{n-1\}}(p)$ by induction on $n$. By Claim \ref{claim_symm}, we have $\mathbf{I}^S_{S\cup\{n-1\}}(p)=f_{S\cup\{n-1\}}(p)=f_{S\cup\{n\}}(p)=\mathbf{I}^S_{S\cup\{n\}}(p)$, and hence, $\mathbf{I}^S_{S\cup\{i\}}(p)\leq\mathbf{I}^S_{S\cup\{n\}}(p)$, as required.
\end{proof}

\section{Proof of Lemma \ref{lem_supermodularity}}

\begin{proof}
Let $S$ be a finite subset of $\mathbb{N}$, $i,j\notin S$, and $S\cup\{i,j\}\subseteq[n]$. Suppose $p=(p_1,\ldots,p_n)$ with $p_1\geq\cdots\geq p_n>0$. Without loss of generality, assume $i<j$ and therefore $p_i\geq p_j$. Then by Lemma \ref{lem_currents}, we have
\begin{equation}\label{eqn_supermod1}
p_i(\varphi_{S\cup\{i\}}(p)-\varphi_S(p))=\mathbf{I}^S_{S\cup\{i\}}(p)\leq\mathbf{I}^S_{S\cup\{j\}}(p)
\end{equation}
If $m$ is the smallest integer not in $S$, then $m\leq i<j$, and the smallest integer not in $S\cup\{j\}$ is $m$. Hence from Lemma \ref{lem_eqn_varphi_tight}, we have
\begin{equation}\label{eqn_supermod2}
\mathbf{I}^{S\cup\{j\}}_{S\cup\{m,j\}}(p)=1+\sum_{j'\in S\cup\{j\}}\mathbf{I}^{S\cup\{j\}\setminus\{j'\}}_{S\cup\{j\}}(p)=1+\mathbf{I}^S_{S\cup\{j\}}(p)+\sum_{j'\in S}\mathbf{I}^{S\cup\{j\}\setminus\{j'\}}_{S\cup\{j\}}(p)\geq\mathbf{I}^S_{S\cup\{j\}}(p)
\end{equation}
where the inequality follows from (\ref{eqn_f_monotone}). By Lemma \ref{lem_currents}, we have
\begin{equation}\label{eqn_supermod3}
p_i(\varphi_{S\cup\{i,j\}}(p)-\varphi_{S\cup\{j\}}(p))=\mathbf{I}^{S\cup\{j\}}_{S\cup\{i,j\}}(p)\geq\mathbf{I}^{S\cup\{j\}}_{S\cup\{m,j\}}(p)
\end{equation}
Putting together (\ref{eqn_supermod1}), (\ref{eqn_supermod2}), and (\ref{eqn_supermod3}), we get
\[p_i(\varphi_{S\cup\{i\}}(p)-\varphi_S(p))\leq p_i(\varphi_{S\cup\{i,j\}}(p)-\varphi_{S\cup\{j\}}(p))\]
and since $p_i>0$,
\[\varphi_{S\cup\{i\}}(p)+\varphi_{S\cup\{j\}}(p)\leq\varphi_{S\cup\{i,j\}}(p)+\varphi_S(p)\]
as required.
\end{proof}

\end{document}